\documentclass[twocolumn]{autart}    %

\usepackage{graphicx}

\usepackage{subcaption}        %
\usepackage{float}             %
\usepackage{tikz}
\usepackage{amsmath}
\usetikzlibrary{positioning, arrows.meta, calc, fit, backgrounds}

\usepackage{mathptmx}          %
\usepackage{amsfonts}          %
\DeclareMathAlphabet{\mathcal}{OMS}{cmsy}{m}{n} %

\usepackage{mathtools}         %
\DeclarePairedDelimiter{\norm}{\lVert}{\rVert}
\usepackage{amssymb}           %
\usepackage{bm}                %

\usepackage{booktabs}          %
\usepackage{multirow}          %
\usepackage[table]{xcolor}     %
\newcommand{\bestres}{green!15} %
\usepackage{makecell}

\usepackage[inline]{enumitem}  %

\usepackage{algorithmic}       %

\usepackage{textcomp}          %
\usepackage{siunitx}           %
\newtheorem{remark}{Remark}
\newtheorem{example}{Example}
\newtheorem{theorem}{Theorem}
\newenvironment{proof}[1][Proof]{\noindent\textbf{#1} }{\qed\par\vspace{1ex}}
\newtheorem{definition}{Definition}
\newtheorem{proposition}{Proposition}
\newtheorem{corollary}{Corollary}
\newtheorem{lemma}{Lemma}
\newtheorem{assumption}{Assumption}
\renewcommand{\qed}{\unskip\nobreak\hfill $\square$} %

\usepackage[draft]{hyperref}   %
\newcommand{\aref}[1]{App.~\ref*{#1}} %

\newcommand{\R}{\mathbb{R}}

\newcommand{\params}{\theta}                    %
\newcommand{\dimParams}{{d_\theta}}             %
\newcommand{\trueOp}{\mathbf{S}}            %
\newcommand{\trueSys}{\mathbf{G}}           %
\newcommand{\trueSysFunc}[1]{G_{#1}}            %
\newcommand{\cont}{\mathbf{K}}              %
\newcommand{\contFunc}[1]{K_{#1}}               %
\newcommand{\approxOp}{\hat{\trueOp}}           %
\newcommand{\approxOpFunc}[1]{\hat{S}_{#1}}     %
\newcommand{\approxSys}{\hat{\trueSys}}         %

\newcommand{\trueIn}[1]{u_{#1}}                 %
\newcommand{\trueInMat}{\mathbf{u}}             %
\newcommand{\trueInData}{\{\trueIn{0:\horizon-1}^\counter\}_{\counter=0}^{\numSamples-1}}
\newcommand{\trueNoisyOut}[1]{y_{#1}}

\newcommand{\trueNoisyOutMat}{\mathbf{y}}

\newcommand{\exc}[1]{r_{#1}}
\newcommand{\excMat}{\mathbf{r}}

\newcommand{\noise}[1]{v_{#1}}
\newcommand{\noiseMat}{\mathbf{v}}

\newcommand{\trueState}[1]{x_{#1}}
\newcommand{\trueStateMat}{\mathbf{x}}

\newcommand{\approxNoisyOut}[1]{\hat{y}_{#1}}
\newcommand{\approxNoisyOutMat}{\hat{\mathbf{y}}}

\newcommand{\whatMat}{\hat{\bm{\omega}}}

\newcommand{\approxNoiseMat}{\mathbf{v}}

\newcommand{\excVec}[1]{r_{#1}}
\newcommand{\noiseVec}[1]{v_{#1}}
\newcommand{\whatVec}[1]{\hat{\omega}_{#1}}

\newcommand{\upOp}{\bm{\Upsilon}}
\newcommand{\upOpKnot}{\bm{\Upsilon^o}}
\newcommand{\upFunc}[1]{\Upsilon_{#1}}
\newcommand{\upKnotFunc}[1]{\Upsilon^o_{#1}}
\newcommand{\eyeOp}{\mathbf{I}}
\newcommand{\aMat}{\mathbf{a}}
\newcommand{\bMat}{\mathbf{b}}
\newcommand{\aVec}[1]{a_{#1}}
\newcommand{\bVec}[1]{b_{#1}}

\newcommand{\dimIn}{{d_u}}
\newcommand{\dimOut}{{d_y}}
\newcommand{\dimState}{{d_x}}

\newcommand{\counter}{n}
\newcommand{\numSamples}{N}
\newcommand{\horizon}{T}
\newcommand{\dataset}{\mathbb{S}}
\newcommand{\realdataset}{\mathcal{D}}

\newcommand{\eye}{I}
\newcommand{\zeroMat}{\mathbf{0}}

\DeclareMathOperator*{\argmin}{arg\,min}
\DeclareMathOperator*{\E}{\mathbb{E}}

\renewcommand{\aref}[1]{\hyperref[#1]{App.~\ref*{#1}}}

\newcommand{\lp}{\ell_p}
\newcommand{\Lp}{\mathcal{L}_p}

\definecolor{darkgray}{HTML}{314354}

\DeclareMathAlphabet\mathbfcal{OMS}{cmsy}{b}{n}

\newcommand{\thetaset}{\Theta}

\renewcommand{\P}{\mathcal{P}}
\newcommand{\Q}{\mathcal{Q}}

\DeclareMathOperator*{\KL}{\mathrm{KL}}
\newcommand{\contextMat}{\boldsymbol{\zeta}}
\newcommand{\dimContext}{{d_\zeta}}

\newcommand{\CL}[2]{\mathrm{CL}[#1,#2]}
\newcommand{\stabset}{\mathcal{G}_{\cont}}       %
\newcommand{\approxSysICI}{\hat{\trueSys}_{\text{ICI}}}   %
\DeclareMathOperator{\sech}{sech}

\begin{document}

\begin{frontmatter}

\title{Neural Closed-Loop Identification of Nonlinear Systems with Stability and Generalization Guarantees\thanksref{footnoteinfo}} %

\thanks[footnoteinfo]{This work was supported as a part of NCCR Automation, a National Centre of Competence in Research, funded by the Swiss National Science Foundation (grant number 51NF40\_225155) and the NECON project (grant number 200021\_219431)}

\author[EPFL]{Laura Meroi}\ead{laura.meroi@epfl.ch},    %
\author[EPFL]{Mahrokh G. Boroujeni}\ead{mahrokh.ghoddousiboroujeni@epfl.ch},               %
\author[EPFL]{Leonardo Massai}\ead{l.massai@epfl.ch},
\author[EPFL]{Giancarlo Ferrari-Trecate}\ead{giancarlo.ferraritrecate@epfl.ch}  %

\address[EPFL]{Institute of Mechanical Engineering, EPFL, Switzerland.}  %


\begin{abstract}                          %

Neural networks are expressive models of nonlinear dynamics, but identifying such systems from closed-loop data raises two difficulties: the feedback loop correlates inputs with measurement noise, biasing standard open-loop estimators, and the identified model must be stabilized by the data-collection controller, so that its closed-loop behavior matches the true plant's. We derive a complete parameterization
of all nonlinear systems stabilized by a given controller, generalizing the
classical dual-Youla parameterization to the nonlinear setting. Building on
this characterization, we propose Internal Controller Identification (ICI), a
computationally tractable framework for identifying potentially unstable
systems while ensuring closed-loop stability by construction; in the linear case, our indirect method reduces to
an open-loop identification problem, for which standard consistency results
apply. The framework accommodates contextual information, letting a single certified model identify a family of systems, and supports a PAC-Bayesian training objective that, when data are scarce and prior knowledge is available, bounds the identified operator's expected prediction loss. Experiments show that ICI identifies open-loop unstable plants where uncertified baselines
diverge, outperforms direct methods under weak excitation and
feedback-correlated noise, generalizes zero-shot to environments unseen
during training, and, with limited data, avoids the overfitting of
empirical risk minimization while providing a non-vacuous generalization bound.

\end{abstract}

\end{frontmatter}

\section{Introduction}\label{sec:intro}

Neural networks are flexible, expressive models for learning nonlinear dynamical systems from input-output data~\cite{Sjoberg99NN,Pillonetto25NN}. However, open-loop data collection is often impractical or unsafe: impossible for unstable systems, which can't be run without stabilizing feedback, and undesirable for stable ones that operate under feedback for safety, setpoint tracking, or disturbance rejection (e.g., drones or robotic manipulators). In these scenarios, data are gathered while the system operates in closed loop with a stabilizing controller, motivating closed-loop system identification~\cite{Vandenhof98CLissues,Forssell99Revisited,Maruta21ASF,Jones22nonasymp}. Compared to the open-loop case, the central difficulty of this setting is the feedback-induced correlation between input signals and measurement noise~\cite{Vandenhof98CLissues,Forssell99Revisited}.
Classical methods for closed-loop identification are typically classified into \textit{direct} and \textit{indirect} approaches~\cite{Soderstrom88SysID}.

\textit{Direct methods} estimate the system and noise model independently of the feedback; in the linear case they are consistent provided the noise model is accurate, an assumption often violated by the complex disturbances of real-world systems~\cite{Vandenhof98CLissues,Forssell99Revisited}. 

\textit{Indirect methods} instead exploit knowledge of the controller and target the closed-loop mapping, eliminating the noise contribution at the plant input and reducing the problem to open-loop identification~\cite{Vandenhof98CLissues,Forssell99Revisited,Ljung99Alternative}; in the linear case this yields consistency even under a misspecified noise model, and it is the route we follow.

A third family, \textit{joint input-output} methods, simultaneously estimates the plant and controller from data jointly modeled as the output of a white-noise-driven system~\cite{Soderstrom88SysID}; however, these methods are largely confined to the linear setting, do not extend naturally to general nonlinear plants, and additionally exhibit theoretically unbounded variance in the small-noise regime~\cite{Abdalmoaty24SmallNoise}.

For LTI systems, one commonly used framework for indirect identification is the dual-Youla parameterization. This formulation is constructed using coprime factorizations of the known controller and of the nominal plant stabilized by it, augmented by a Youla parameter that is optimized to minimize the discrepancy from the true system. Through this structure, the parameterization characterizes the complete set of all plants stabilized by the given controller~\cite{Hensen89DualYoula}; as its name suggests, it is dual to the Youla-Kučera parameterization, which instead characterizes all controllers that stabilize a given plant~\cite{Youla76Youla}. A fundamental advantage of this approach is that any identified model is inherently stabilized by the controller, naturally preserving closed-loop stability properties observed in the true system~\cite{Forssell99Revisited}. Furthermore, in the linear case, the Youla parameter reduces to a stable rational transfer function; since the underlying function space is infinite-dimensional, it is truncated to a finite order in practice, yielding a tractable optimization.

Extensions to nonlinear plants under linear control exist (see~\cite{anderson1998youla} for a review, and~\cite{dasgupta1996parametrization,fujimoto2001kernel} for coprime-factorization parameterizations) but face two obstacles: coprime factorizations exist only for subclasses of nonlinear systems and are hard to compute~\cite{anderson1998youla}, and the parameterization requires optimizing over stable nonlinear operators and their inverses, which complicates the optimization, a difficulty also noted for Youla-based nonlinear controller learning~\cite{outputSLS}. Existing methods thus fall short of a computationally practical framework.

More recently,~\cite{Maruta21ASF} proposed an alternative approach for closed-loop identification of nonlinear systems. While effective in practice, this method has two key limitations: 
\begin{enumerate*}[label=(\roman*)] 
    \item the closed-loop system of the learned model and controller may become unstable, and 
    \item it assumes a predefined parametric form for the system, limiting flexibility. 
\end{enumerate*}
This leaves closed-loop identification of expressive, nonlinear models a critical, yet not sufficiently explored, challenge.

A further challenge is that real-world systems rarely operate in a single
environment: friction,
payload, or terrain vary, so the data are generated by a family of related plants. Identification across operating conditions is classically
addressed by linear parameter-varying (LPV) methods, also in closed
loop~\cite{van2009subspace,toth2010modeling}, but restricted to LPV model
classes without closed-loop stability guarantees; meta- and in-context
learning estimate a single meta-model for a class of
systems~\cite{forgione2023system,piga2024adaptation}, again without
certificates.

A distinct concern is finite-sample reliability: closed-loop data
collection is costly, safety-critical, and biased by the controller, so
guarantees from limited data are valuable. Such guarantees can be obtained
through PAC-Bayes theory~\cite{mcallester1999some,userfriendly}, a powerful
tool for certifying highly parameterized models, applied beyond supervised
learning to meta-learning~\cite{PACOH} and optimal
control~\cite{Boroujeni24PAC}. PAC-Bayesian bounds exist for open-loop
identification of stochastic LTI
systems~\cite{eringis2021pac,eringis2024pac} and stable recurrent
models~\cite{eringis2023pac}; their use in closed-loop
identification, under strict closed-loop stability constraints on the
learned model, remains unexplored.

\textbf{Contributions.}
We propose a computationally tractable framework for closed-loop identification of nonlinear, potentially open-loop unstable systems while guaranteeing stability of the interconnection with the data-collection controller. 
In particular:
\begin{enumerate*}[label=(\roman*)]
\item we introduce the Internal Controller Identification (ICI)
framework, a complete parameterization of all and only the nonlinear systems stabilized by a given controller, which generalizes the dual-Youla
parameterization, guarantees closed-loop stability by construction, and reduces to a consistent open-loop problem in the linear case;
\item we derive finite-sample generalization guarantees, introducing an informative prior over the parameters of the learnable operator and a tractable PAC-Bayesian bound on its expected prediction loss, optimized with
Stein Variational Gradient Descent (SVGD);
\item we develop context-aware identification, where bounded context wrappers let a single model identify an entire family of environments and generalize to members never observed during training, with the closed-loop stability certificate intact;
\item we conduct a systematic experimental evaluation covering open-loop unstable plants, weak excitation, feedback-correlated noise, and partial observation, against closed-loop baselines.
\end{enumerate*}

A preliminary version of this work appeared in the conference
paper~\cite{Boroujeni2025Neural}, which introduced the base ICI parameterization and its stability guarantee. The present paper adds the context-aware formulation (Section~\ref{sec:context_ici}), the PAC-Bayesian
generalization framework and its SVGD-based training
(Section~\ref{sec:PAC}), the four-experiment study
(Section~\ref{sec:experiments}), and complete proofs.

\subsection{Notation}
The set of all sequences $\mathbf{x} = (x_0,x_1,x_2,\ldots)$, where $x_t \in \mathbb{R}^\dimState$, $t\in \mathbb{N}$, is denoted as $\ell^\dimState$. 
Moreover,  $\mathbf{x}$ belongs to $\ell_p^\dimState \subset \ell^\dimState$ with $p \geq 1$ if $\norm{\mathbf{x}}_p = \left(\sum_{t=0}^\infty |x_t|^p\right)^{\frac{1}{p}} < \infty$ and $\mathbf{x} \in \ell^\dimState_\infty$ if $\operatorname{sup}_{t}|x_t|< \infty$, where $|\cdot|$ denotes any vector norm. 
We use the notation $x_{i:j}$ to refer to the truncation of $\mathbf{x}$ to the finite-dimensional vector $(x_i,x_{i+1},\ldots,x_{j})$ for $j \geq i$, and $\trueState{i:j} =\varnothing$ otherwise.
An operator $\upOp: \ell^\dimState \to \ell^\dimOut$ is a mapping between two sequences and it is \textit{causal} if 
$\upOp(\trueStateMat) = \bigl(\upFunc{0}(\trueState{0}), \dots, \upFunc{t}(\trueState{0:t}), \dots\bigr)$, and \textit{strictly causal} if $\upOp(\trueStateMat) = \bigl(\upFunc{0}(\varnothing), \dots, \upFunc{t}(\trueState{0:t-1}), \dots\bigr)$, where $\upFunc{0}(\varnothing)$ is a vector in $\R^\dimOut$.
A causal operator $\upOp$ is:
\begin{itemize}
    \item $\ell_p$-stable ($\upOp \in \Lp$) if $\upOp(\trueStateMat) \in \lp^\dimOut$ for all $\trueStateMat \in \lp^\dimState$.
    \item incrementally $\ell_p$-stable if $\exists \, \gamma \in [0, \infty)$ such that $\forall \, \trueStateMat_1, \trueStateMat_2 \in \lp^\dimState$, 
        $\Vert \upOp(\trueStateMat_1) - \upOp(\trueStateMat_2)\Vert_p \leq \gamma \, \Vert\trueStateMat_1 - \trueStateMat_2\Vert_p$,
    where $\gamma$ is the incremental finite gain (i.f.g.) of $\upOp$.
\end{itemize}
We define $\upOp_{i:j}(x_{0:j}) =  (\upOp_i(x_{0:i}),\upOp_{i+1}(x_{0:i+1}),\ldots,\upOp_j(x_{0:j}))$; for strictly causal $\upOp$, $\upOp_{i:j}$ acts on $x_{0:j-1}$. 
The sets of causal and strictly causal $\Lp$-stable operators are denoted $\Lp^C$ and $\Lp^{SC}$, respectively.

\section{Problem formulation} \label{sec:problem_formulation}
\subsection{Closed-loop system dynamics}
    We consider a discrete-time dynamical system described by the input-output model:
    \begin{align} 
        \begin{cases} 
            y_0 = G_0(\varnothing) + v_0, 
            \\ 
            y_t = G_t(u_{0:t-1}) + v_t, 
            \qquad t \in \mathbb{N},  
        \end{cases} \label{eq:true_sys}
    \end{align}
    where $G_0(\varnothing) \in \mathbb{R}^{d_y}$ is the nominal initial output. 
    At each time $t \in \mathbb{N}\cup \{0\}$, $u_t \in \mathbb{R}^{d_u}$ denotes the input, $y_t \in \mathbb{R}^{d_y}$ is the output, and $v_t \in \mathbb{R}^{d_y}$ is the measurement noise corrupting the output. 
    The dynamics $G_t$ may be time-varying, nonlinear, and potentially unstable; by~\eqref{eq:true_sys}, $y_t$ depends only on $u_{0:t-1}$, i.e., the system is strictly causal. In operator notation, system~\eqref{eq:true_sys} reads:
\begin{align}
    \trueNoisyOutMat = \trueSys (\trueInMat) + \noiseMat, \label{eq:true_sys_op}
\end{align}
where $\trueSys$ is a strictly causal operator such that $\trueSys (\trueInMat) = \bigl( \trueSysFunc{0} (\varnothing), \trueSysFunc{1} (\trueIn{0}), \dots\bigr)$, $\trueInMat = (\trueIn{0}, \trueIn{1}, \dots)$ and $\noiseMat$, $\trueNoisyOutMat$ are defined in a similar way. The next example shows the relation between~\eqref{eq:true_sys} and state-space models.
\begin{example}[From state-space to input-output models]
\label{ex:nonlin}
Consider a strictly causal, time-invariant system in the state-space form:
\begin{align}
    \begin{cases} 
        \trueState{t+1} = f(\trueState{t}, \trueIn{t}) \,, \quad \trueState{0} = \bar{\trueState{}},\\
        \trueNoisyOut{t} = h(\trueState{t}) + \noise{t},
    \end{cases}\label{eq:state_space}
\end{align}
where $\trueState{t} \in \R^\dimState$ is the state at time $t$ and $\bar{\trueState{}} \in \R^\dimState$ is the initial state.
By recursion, the functions $\trueSysFunc{t}$ in~\eqref{eq:true_sys} are defined as: 
   \begin{align*}
        \trueSysFunc{0}(\varnothing) &= h \bigl(\bar{\trueState{}}\bigr) , \; \trueSysFunc{1}(\trueIn{0}) = h \bigl( f (\bar{\trueState{}}, \trueIn{0})\bigr) , \; \dots\\
        \trueSysFunc{t}(\trueIn{0:t-1}) &= h \bigl(\underbrace{ f ( \dots ( f }_{t \textit{ times}}(\bar{\trueState{}}, \trueIn{0}), \trueIn{1}), \dots, \trueIn{t-1})\bigr) \,, \; 
        \dots \;.
    \end{align*} 
Note that $\trueSysFunc{t}$ depends on the initial state $\bar{\trueState{}}$ for all $t$, consistent with input-output models in the literature~\cite{paice1996tac}. 
When the system~\eqref{eq:state_space} is LTI and $\bar{\trueState{}} = 0$, $\trueSys$ reduces to a linear time-invariant convolution operator, equivalently described in the $\mathcal{Z}$-domain by the strictly proper transfer matrix $\trueSys(z)$.
\end{example}

Next, we introduce the closed-loop system, illustrated in Figure~\ref{fig:true_CL} and described by:
\begin{align} 
    \CL{\trueSys}{\cont} : \quad
    \begin{cases} 
        \trueNoisyOutMat = \trueSys (\trueInMat) + \noiseMat,
        \\ 
        \trueInMat =  \cont(\trueNoisyOutMat) +\excMat ,  
    \end{cases} \label{eq:true_cl}
\end{align}
where $\cont$ is a causal, time-varying dynamical output-feedback controller and $\mathbf{r} = (r_0, r_1, \dots)$, $\exc{t} \in \R^\dimIn$, is an external excitation injected in addition to the feedback. In~\eqref{eq:true_cl}, $\CL{\trueSys}{\cont}$ denotes the operator $(\excMat, \noiseMat) \mapsto (\trueInMat, \trueNoisyOutMat)$ that the two equations define implicitly, i.e., the closed loop of Figure~\ref{fig:true_CL} seen as a single system mapping the exogenous signals to the internal ones. The same notation is used for any strictly causal system in place of $\trueSys$.

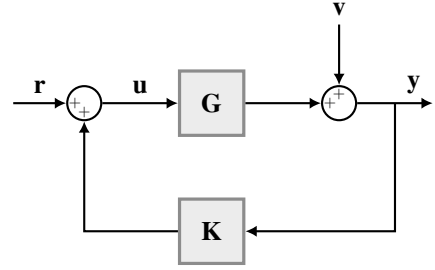
\begin{figure}[htbp]
    \centering
    \begin{tikzpicture}[
        auto,
        block/.style = {draw=gray!90, rectangle, minimum height=0.85cm, minimum width=0.85cm, align=center, line width=1.2pt, fill=gray!15},
        sum/.style = {draw, circle, minimum size=0.45cm, inner sep=0pt, line width=0.8pt},
        line/.style = {draw, -{Latex[length=1.5mm, width=1.5mm]}, line width=0.8pt},
        font=\normalsize\bfseries
    ]

    \node[coordinate] (input) {};
    \node[sum, right=0.7cm of input] (sum1) {};
    
    \node[block, right=1cm of sum1] (G) {G};
    \node[sum, right=1cm of G] (sum2) {};
    \node[coordinate, above=0.8cm of sum2] (dist) {};
    \node[coordinate, right=1cm of sum2] (output) {};
    
    \node[block, below=0.8cm of G] (K) {K};

    \node[font=\tiny] at ([xshift=-0.12cm]sum1.center) {$+$};
    \node[font=\tiny] at ([yshift=-0.12cm]sum1.center) {$+$};

    \node[font=\tiny] at ([xshift=-0.12cm]sum2.center) {$+$};
    \node[font=\tiny] at ([yshift=0.12cm]sum2.center) {$+$};

    \path[line] (input) -- node[above] {r} (sum1.west);
    \path[line] (sum1.east) -- node[above] {u} (G.west);
    \path[line] (G.east) -- (sum2.west);
    
    \path[line] (dist) -- (sum2.north);
    \node[above] at (dist) {v};
    
    \path[line] (sum2.east) -- node[above, near end] {y} (output);

    \coordinate (tap) at ($(sum2.east) + (0.5cm,0)$); 
    \path[line] (tap) |- (K.east); 
    \path[line] (K.west) -| (sum1.south); 

    \end{tikzpicture}
    \caption{Block diagram of the closed-loop system formed by system $\trueSys$ and controller $\cont$, with system input $\trueInMat$, noisy output $\trueNoisyOutMat$, excitation signal $\excMat$, and noise $\noiseMat$.}
    \label{fig:true_CL}
\end{figure}

As discussed in Section~\ref{sec:intro}, data are collected while the system operates in closed loop. To ensure reliable data collection and prevent unbounded measurements, the closed-loop system must satisfy an appropriate notion of stability. 

\begin{definition}[Closed-loop $\ell_p$-stability~\cite{PerfBoost}]\label{def:stability}

    Let $\trueSys$ be a strictly causal operator and $\cont$ a causal operator. The closed loop of $\trueSys$ and $\cont$ in~\eqref{eq:true_cl} is $\ell_p$-stable if $\trueInMat \in \lp$ and $\trueNoisyOutMat \in \lp$ for every $\excMat \in \lp$ and every $\noiseMat \in \lp$. Equivalently, we write $\CL{\trueSys}{\cont} \in \Lp$. We denote by $\stabset := \{\, \trueSys \text{ strictly causal} : \CL{\trueSys}{\cont} \in \Lp \,\}$ the set of systems stabilized by $\cont$.
\end{definition}

Definition~\ref{def:stability} constrains only the interconnection: neither $\trueSys$ nor $\cont$ is required to belong to $\Lp$ individually. When $\trueSys$ and $\cont$ are LTI, Definition~\ref{def:stability} recovers the standard notion: $\CL{\trueSys}{\cont} \in \Lp$ for every $p \geq 1$ if and only if the transfer matrix describing $\CL{\trueSys}{\cont}$ is asymptotically stable.

\subsection{Identification Objective}
Our goal is to learn a dynamical model $\approxSys$ that best approximates the true system $\trueSys$ using closed-loop data collected with a known stabilizing controller $\cont$. 
While $\approxSys$ may be open-loop unstable, capturing potential plant instability, we enforce $\CL{\approxSys}{\cont} \in \Lp$, i.e., we constrain the search to $\stabset$.

The stochasticity of the closed-loop data stems from the unobserved measurement noise $\mathbf{v}$ and from the excitation $\excMat$, actively sampled from a user-designed distribution to make the input informative. Since the system operates under feedback, the control input $\mathbf{u}$ is a deterministic function of  the measured signals $(\excMat,\mathbf{y})$ through the closed-loop equations~\eqref{eq:true_cl}. Consequently, a data sample from a given experiment is the fully observable, infinite-length trajectory $\bm{\tau} = (\mathbf{r},\mathbf{u},\mathbf{y})$, which is drawn from the underlying closed-loop data distribution $\mathcal{D}$ induced by the distributions of $\excMat$ and $\mathbf{v}$.

Ideally, we seek to minimize an expected identification loss $c$ over $\mathcal{D}$, as follows:
\begin{align} 
    \min_{\approxSys}\;& \mathbb{E}_{\bm{\tau}\sim \mathcal{D}} 
    \; c(\approxSys,\bm{\tau})
    \label{eq:ideal_objective} \\
    \text{s.t.} \;& 
    \CL{\approxSys}{\cont} \in \Lp , \nonumber 
\end{align}
where $c\big(\hat{\mathbf{G}},\bm{\tau}\big)$ quantifies the mismatch between the
predicted output $\approxNoisyOutMat$ and the measured $\mathbf{y}$.
Importantly, \eqref{eq:ideal_objective} does not commit to how $\approxNoisyOutMat$ is generated from $\bm{\tau}$: this choice, deferred to
Section~\ref{sec:ERM_framework}, gives rise to the direct and indirect
strategies.

Directly solving the optimization problem in~\eqref{eq:ideal_objective} is impractical due to three fundamental challenges:
\begin{enumerate}[label=(C\arabic*)]
    \item \label{challenge:constraints} \textbf{Complex constraints:} The set $\stabset$ is generally non-convex and difficult to parameterize directly.
    \item \label{challenge:horizon} \textbf{Intractable horizon:} Evaluating the objective over infinite-length trajectories is computationally infeasible.
    \item \label{challenge:distribution} \textbf{Unknown distribution:} The true data distribution $\mathcal{D}$ is unknown, preventing the direct computation of the expectation.
\end{enumerate}

In the following, we address these challenges systematically. We first focus on the stability constraint~\ref{challenge:constraints} via the ICI framework described in Section~\ref{subsec:ici}. Subsequently in Section~\ref{sec:ERM_framework} we tackle challenges~\ref{challenge:horizon} and ~\ref{challenge:distribution}, by adopting an Empirical Risk Minimization (ERM) approach defined over finite-length trajectories.

\section{Internal Controller Identification (ICI) framework}\label{subsec:ici}
\subsection{Parameterization of Stabilized Systems}
Problem~\eqref{eq:ideal_objective} requires learning over $\stabset$, a constraint that is hard to enforce directly because this set is generally non-convex. We parameterize \emph{all and only} such systems through a dual formulation of Internal Model Control (IMC) architectures~\cite{IMC}, which we call the Internal Controller Identification (ICI) framework. Specifically, $\approxSysICI$ is the negative feedback interconnection in Figure~\ref{fig:Ghat} of a copy of the controller $\cont$ and a trainable strictly causal operator $\approxOp \in \Lp^{SC}$, so that, for an input $\hat{\trueInMat}$, the predicted output satisfies
\begin{equation}
    \approxNoisyOutMat = \approxSysICI(\hat{\trueInMat}) = \approxOp \bigl( \hat{\trueInMat} - \cont (\approxNoisyOutMat) \bigr). \label{eq:ici}
\end{equation}
This parameterization transforms the problem of learning $\approxSysICI$ into optimizing $\approxOp$ over the set $\Lp^{SC}$~\cite{outputSLS,PerfBoost}.

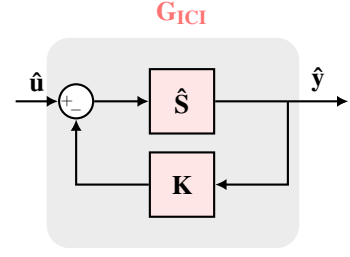
\begin{figure}[htbp]
    \centering
    \begin{tikzpicture}[
        auto,
        block/.style = {draw, rectangle, minimum height=0.85cm, minimum width=0.85cm, align=center, line width=0.8pt, fill=pink!40},
        sum/.style = {draw, circle, minimum size=0.45cm, inner sep=0pt, line width=0.8pt, fill=white},
        line/.style = {draw, -{Latex[length=1.5mm, width=1.5mm]}, line width=0.8pt},
        font=\normalsize\bfseries
    ]

    \node[sum] (sum1) at (0, 0) {};
    \node[block] (Shat) at (1.4, 0) {$\mathbf{\hat{S}}$};
    \node[block] (K) at (1.4, -1.1) {$\mathbf{K}$}; 
    
    \coordinate (tap) at (2.8, 0);

    \node[coordinate] (input) at (-0.8, 0) {};
    \node[coordinate] (output) at (3.6, 0) {};

    \node[font=\tiny] at ([xshift=-0.12cm]sum1.center) {$+$};
    \node[font=\tiny] at ([yshift=-0.12cm]sum1.center) {$-$};

    \path[line] (input) -- node[above] {$\mathbf{\hat{u}}$} (sum1.west);
    \path[line] (sum1.east) -- (Shat.west);
    
    \path[draw, line width=0.8pt] (Shat.east) -- (tap); 
    \path[line] (tap) -- node[above] {$\mathbf{\hat{y}}$} (output); 
    
    \path[line] (tap) |- (K.east); 
    \path[line] (K.west) -| (sum1.south); 

    \begin{scope}[on background layer]
        \fill[gray!15, rounded corners=3mm] 
            ([xshift=-0.15cm, yshift=0.4cm]sum1.west |- Shat.north) 
            rectangle 
            ([xshift=0.15cm, yshift=-0.4cm]tap |- K.south);
    \end{scope}

    \node[font=\normalsize\bfseries, text=pink!60!red, anchor=south] at ([yshift=0.45cm]Shat.north) {$\mathbf{\hat{G}}_{\text{ICI}}$}; 

    \end{tikzpicture}
    \caption{In the ICI framework, the model $\approxSysICI$ consists of the negative feedback interconnection of a copy of the controller $\cont$ and a trainable operator $\approxOp \in \Lp^{SC}$. The model input and output are denoted by $\hat{\mathbf{u}}$ and $\approxNoisyOutMat$, respectively.}
    \label{fig:Ghat}
\end{figure}

The next theorem relates $\stabset$ to the models representable in the ICI framework. Note that statements~\ref{item:T1} and~\ref{item:T2} are not symmetric: sufficiency requires $\cont \in \Lp^\textit{C}$ and incrementally
$\ell_p$-stable, whereas necessity uses only causality of $\cont$. Note that the second hypothesis of~\ref{item:T1} subsumes the first whenever $\cont(\zeroMat)\in\lp$.
\begin{theorem}\label{theo:ici}
    Let $\cont$ be causal and let $\approxOp \in \Lp^\textit{SC}$. Let $\approxSysICI$ be the ICI model defined by~\eqref{eq:ici}.
    \begin{enumerate}[label=\textcolor{darkgray}{\textbf{T.\arabic*}}]
        \item \label{item:T1} \emph{(Sufficiency.)} If $\cont \in \Lp^\textit{C}$ and $\cont$ is incrementally $\ell_p$-stable, then $\approxSysICI$ is strictly causal and $\CL{\approxSysICI}{\cont} \in \Lp$, i.e., $\approxSysICI \in \stabset$.
        \item \label{item:T2} \emph{(Necessity.)} Conversely, for every $\trueSys \in \stabset$ there exists a unique $\approxOp \in \Lp^\textit{SC}$ such that the ICI model~\eqref{eq:ici} satisfies $\approxSysICI = \trueSys$.
    \end{enumerate}
\end{theorem}
\begin{corollary}\label{coroll:bijection}
    Under the hypotheses of~\ref{item:T1}, the map $\approxOp \mapsto \approxSysICI$ defined by~\eqref{eq:ici} is a bijection from $\Lp^\textit{SC}$ onto $\stabset$.
\end{corollary}

\begin{proof}[Proof of Theorem \ref{theo:ici}]
    We first state a useful result from~\cite{Ho20Corol}.
\begin{lemma}[\cite{Ho20Corol}]\label{coroll:inverse}
Consider a causal operator $\upOp$ such that $\upOpKnot := \upOp - \eyeOp$ is strictly causal, where $\eyeOp$ is the identity operator. Then, the causal operator $\upOp^{-1}$ exists, and the sequence $\bMat = \upOp^{-1} \aMat$ satisfies
\begin{equation}\label{eq:lemma_recursion}
    \bVec{t} = \aVec{t} - \upKnotFunc{t} (\bVec{0:t-1}), \qquad t = 0, 1, \dots
\end{equation}
\end{lemma}

We start by proving~\ref{item:T1}: for every $\excMat \in \lp^\dimIn$ and $\noiseMat \in \lp^\dimOut$ we show that both $\approxNoisyOutMat \in \lp^\dimOut$ and $\hat{\trueInMat} \in \lp^\dimIn$, so that $\CL{\approxSysICI}{\cont} \in \Lp$. Denote by $\whatMat$ the input to the operator $\approxOp$. In $\CL{\approxSysICI}{\cont}$ the external copy of $\cont$ is driven by the measured output $\approxNoisyOutMat + \noiseMat$, while the copy internal to $\approxSysICI$ is driven by $\approxNoisyOutMat$; hence
\begin{align}
    \whatMat = \hat{\trueInMat} - \cont (\approxNoisyOutMat) = \excMat + \cont \bigl(\approxOp(\whatMat) + \noiseMat \bigr) - \cont \bigl(\approxOp(\whatMat) \bigr).
    \label{eq:proof_what}
\end{align}

Since $\approxOp$ is strictly causal and $\cont$ is causal, $\whatVec{t}$ in~\eqref{eq:proof_what} depends only on $\excVec{0:t}$, $\noiseVec{0:t}$ and $\whatVec{0:t-1}$; hence $\whatMat$ exists and is unique, so the loop is well posed, and the same recursion shows that $\approxSysICI$ is strictly causal.

Let $P_M$ denote the truncation operator defined as $(P_M \mathbf{x})_t = x_t$ for $t \leq M$ and $(P_M \mathbf{x})_t = 0$ otherwise. 

Every truncated signal belongs to $\lp$, and causality of $\cont$ gives $P_M \cont(\mathbf{x}) = P_M \cont(P_M \mathbf{x})$. Since $P_M \mathbf{x} \in \lp$ for every $\mathbf{x}$, the i.f.g.\ inequality $\Vert \cont(\mathbf{x}_1) - \cont(\mathbf{x}_2) \Vert_p \leq \gamma_K \Vert \mathbf{x}_1 - \mathbf{x}_2 \Vert_p$ holds for the truncated signals~\cite{van2000l2}.
Truncating~\eqref{eq:proof_what} and using $\Vert P_M \mathbf{x} \Vert_p \leq \Vert \mathbf{x} \Vert_p$ together with the triangle inequality,
\begin{align*}
    \Vert P_M \whatMat \Vert_p
    &\leq \Vert \excMat \Vert_p + \bigl\Vert \cont \bigl(P_M \approxOp(\whatMat) + P_M \noiseMat \bigr) - \cont \bigl(P_M \approxOp(\whatMat) \bigr) \bigr\Vert_p \\
    &\leq \Vert \excMat \Vert_p + \gamma_K \Vert P_M \noiseMat \Vert_p
    \;\leq\; \Vert \excMat \Vert_p + \gamma_K \Vert \noiseMat \Vert_p ,
\end{align*}
uniformly in $M$, where $\gamma_K$ is the incremental finite gain of $\cont$. 

Since $M \mapsto \Vert P_M \whatMat \Vert_p$ is nondecreasing and bounded above, it converges to $\Vert \whatMat \Vert_p$ as $M \to \infty$; therefore $\whatMat \in \lp^\dimIn$ and $\Vert \whatMat \Vert_p \leq \Vert \excMat \Vert_p + \gamma_K \Vert \noiseMat \Vert_p$.

Since $\approxOp \in \Lp$, it follows that $\approxNoisyOutMat = \approxOp(\whatMat) \in \lp^\dimOut$, and so is the measured output $\approxNoisyOutMat + \noiseMat$. Consequently, since $\cont \in \Lp$, we obtain $\hat{\trueInMat} = \cont (\approxNoisyOutMat + \noiseMat) + \excMat \in \lp^\dimIn$. Therefore $\CL{\approxSysICI}{\cont} \in \Lp$, i.e., $\approxSysICI \in \stabset$.

Next, we prove~\ref{item:T2}. Given $\trueSys \in \stabset$, define $\approxOp := \trueSys(\eyeOp - \cont \trueSys)^{-1}$, where the inverse exists by Lemma~\ref{coroll:inverse} since $\cont \trueSys$ is strictly causal. The operator $\approxOp$ is strictly causal, being the composition of the causal operator $(\eyeOp - \cont \trueSys)^{-1}$ with the strictly causal $\trueSys$; moreover $\approxOp \in \Lp$, since it coincides with the restriction of $\CL{\trueSys}{\cont}$ to $\noiseMat = \zeroMat$, i.e., with the map $\excMat \mapsto \trueNoisyOutMat$, which belongs to $\Lp$ because $\trueSys \in \stabset$ (see also~\cite{Desoer82GlobalPO}). 
Note that this construction gives $\approxOpFunc{0}(\varnothing) = \trueSysFunc{0}(\varnothing)$ automatically.

We apply the same input $\hat{\trueInMat}$ to both $\trueSys$ and $\approxSysICI$, as illustrated in Figure~\ref{fig:Ghat}, and denote their respective outputs by $\approxNoisyOutMat^{\mathrm{des}}$ and $\approxNoisyOutMat$. Our goal is to prove that $\approxNoisyOutMat^{\mathrm{des}} = \approxNoisyOutMat$, or equivalently, $\hat{y}_{t}^{\mathrm{des}} = \hat{y}_{t}$ for all $t \geq 0$.

For $\trueSys$ we have $\approxNoisyOutMat^{\mathrm{des}} = \trueSys (\hat{\trueInMat})$, while from~\eqref{eq:ici},
\begin{align*}
    \approxNoisyOutMat &= \approxSysICI(\hat{\trueInMat}) = \approxOp \bigl( \hat{\trueInMat} - \cont (\approxNoisyOutMat) \bigr)\\
    &= \trueSys \underbrace{(\eyeOp - \cont \trueSys)^{-1}\bigl( \hat{\trueInMat} - \cont (\approxNoisyOutMat) \bigr)}_{\bMat} = \trueSys \, \bMat.
\end{align*}
These relations are equivalent to the causal relations
\begin{align}
    \hat{y}_{t}^{\mathrm{des}} = \trueSysFunc{t} (\hat{u}_{0:t-1}),
    \qquad
    \hat{y}_{t} = \trueSysFunc{t}(\bVec{0:t-1}), \qquad \forall t \geq 0, \label{eq:proof_ici_timedomain}
\end{align}
where we used the strict causality of $\trueSys$ to truncate the operands at time $t-1$. Moreover, by the recursion~\eqref{eq:lemma_recursion} of Lemma~\ref{coroll:inverse},
\begin{align}\label{eq:bt2}
    \bVec{t} = \hat{u}_t - \contFunc{t}(\hat{y}_{0:t}) + \contFunc{t} \bigl( \trueSysFunc{0:t}(\bVec{0:t-1})\bigr), \qquad \forall t \geq 0.
\end{align}
The second relation in~\eqref{eq:proof_ici_timedomain} gives $\hat{y}_{0:t} = \trueSysFunc{0:t}(\bVec{0:t-1})$, so the two controller terms in~\eqref{eq:bt2} cancel identically and $\bVec{t} = \hat{u}_t$ for all $t \geq 0$. Substituting into~\eqref{eq:proof_ici_timedomain} yields $\hat{y}_{t}^{\mathrm{des}} = \hat{y}_{t}$ for all $t \geq 0$.

Uniqueness is forced by~\eqref{eq:ici}. If $\approxOp \in \Lp^{SC}$ satisfies $\approxSysICI = \trueSys$, then $\trueSys(\hat{\trueInMat}) = \approxOp\bigl(\hat{\trueInMat} - \cont(\trueSys(\hat{\trueInMat}))\bigr)$ for every $\hat{\trueInMat}$. Setting $\bMat = (\eyeOp - \cont \trueSys)(\hat{\trueInMat})$ and using the invertibility of $\eyeOp - \cont \trueSys$ from Lemma~\ref{coroll:inverse} gives $\approxOp(\bMat) = \trueSys (\eyeOp - \cont \trueSys)^{-1}(\bMat)$ for every $\bMat$, i.e., $\approxOp = \trueSys(\eyeOp - \cont \trueSys)^{-1}$.
\end{proof}
\begin{proof}[Proof of Corollary~\ref{coroll:bijection}]
By~\ref{item:T1}, the map $\approxOp\mapsto\approxSysICI$ is well defined from $\Lp^\textit{SC}$
into $\stabset$, and existence in~\ref{item:T2} makes it surjective onto $\stabset$. For
injectivity, let $\approxOp_1,\approxOp_2\in\Lp^\textit{SC}$ induce the same model
$\trueSys:=\approxSysICI^{1}=\approxSysICI^{2}$. Then $\trueSys\in\stabset$ by~\ref{item:T1}, so
uniqueness in~\ref{item:T2} gives $\approxOp_1=\approxOp_2$.
\end{proof}

\begin{remark}[Open-loop instability of $\approxSysICI$]
Note that Theorem~\ref{theo:ici} certifies $\CL{\approxSysICI}{\cont} \in \Lp$, but $\approxSysICI$ itself need not be $\lp$-stable: while $\approxOp \in \Lp$ by construction, the corresponding $\approxSysICI$, recovered through the feedback equation~\eqref{eq:ici}, may be open-loop unstable. This is precisely what allows the ICI framework to identify inherently unstable plants while keeping the identified closed-loop behavior consistent with the true one.
\end{remark}

Since $\trueSys \in \stabset$, point~\ref{item:T2} in Theorem~\ref{theo:ici} associates with $\trueSys$ a unique operator $\trueOp \in \Lp^{SC}$, which we call the \emph{true ICI operator}.
By Corollary~\ref{coroll:bijection}, the constrained problem~\eqref{eq:ideal_objective} can be replaced by the following unconstrained optimization over $\approxOp$:
\begin{align}
    \min_{\approxOp \in \Lp^{SC}}\;& \mathbb{E}_{\bm{\tau}\sim \mathcal{D}} 
    \; c(\approxOp,\bm{\tau}),
    \label{eq:ici_risk_min} 
\end{align}
where the associated model $\approxSysICI$ is the interconnection of $\hat{\mathbf{S}}$ and $\mathbf{K}$ defined in~\eqref{eq:ici}.
The key difference is that the constraint $\CL{\approxSysICI}{\cont} \in \Lp$ is now structurally satisfied for every $\approxOp \in \Lp^{SC}$.

\begin{remark}[Recursive implementation of $\approxOp$]\label{remark:recursive}

    Since $\approxOp$ is strictly causal, the output $\hat{y}_t$ depends on past outputs only up to time $t-1$.
        This observation enables a recursive computation of $\hat{y}_t$, eliminating the need to solve the implicit nonlinear equation~\eqref{eq:ici} at each time step.
    Additionally, if  $\approxOp$ and $\cont$ are input-output operators associated with dynamical systems (see Example~\ref{ex:nonlin}), their state-space representations can be used to implement~\eqref{eq:ici} as a Markovian dynamical system, thereby avoiding the need to store the full history of $\hat{u}_{0:t-1}$ and $\hat{y}_{0:t-1}$.
\end{remark}

\subsection{Unconstrained optimization of the operator \texorpdfstring{$\approxOp$}{Lg}}\label{sec:unconstrained}

While the ICI parameterization~\eqref{eq:ici_risk_min} removes the explicit stability constraint, optimizing over the infinite-dimensional space $\Lp^{SC}$ remains intractable. To make this feasible, the operator $\approxOp$ can be modeled using families of $\Lp^\textit{SC}$ operators described by a finite-dimensional parameter vector $\params \in \R^\dimParams$, such as Recurrent Equilibrium Networks (RENs)~\cite{revay2023recurrent}, certain classes of State-Space Models (SSMs)~\cite{gu2022efficiently,massai2025free,zakwan2026controller}, or neural network representations of Hamiltonian systems~\cite{zakwan2024neural}, which admit a state-space form enabling recursive output computation (see Remark~\ref{remark:recursive}). Crucially, these families yield a free parameterization, guaranteeing that the parameterized operator $\approxOp^{\params}$ belongs to $\Lp^\textit{SC}$ for any arbitrary $\params \in \R^\dimParams$.
Defining $\approxSysICI^{\params}$ as the ICI model~\eqref{eq:ici} built from $\approxOp^{\theta}$ and $\cont$, problem~\eqref{eq:ici_risk_min} reduces to a tractable, unconstrained optimization over $\theta \in \mathbb{R}^{d_{\theta}}$:
\begin{equation}\label{eq:ici_obj_theta}
  \min_{\theta\in\mathbb{R}^{d_\theta}}\;
  \mathbb{E}_{\bm{\tau}\sim\mathcal{D}}\,
  c\big(\approxOp^{\theta},\bm{\tau}\big),
\end{equation}
which can be efficiently solved using gradient-based methods.
Since the predictor is uniquely determined by $\theta$, with a small abuse of notation, hereafter we use $c(\theta,\bm{\tau})$ instead of $c(\approxOp^\theta,\bm{\tau})$. 

\subsection{Context-enriched ICI}
\label{sec:context_ici}

Plants often operate under varying conditions, e.g., for robotic systems, changes in payload or terrain-dependent friction. When measurable, we represent information about these conditions by a context sequence $\contextMat=(\zeta_0,\zeta_1,\ldots)$, with $\zeta_t\in\mathcal Z\subset\R^{\dimContext}$. Providing context can help a model adapt its predictions across operating regimes without having to infer those conditions solely from the input history.

Crucially, context may persist even when the dynamical input decays: for example, a constant nonzero payload signal is bounded but does not belong to $\lp$ for $p<\infty$. Simply appending context to the input of an $\lp$-stable operator therefore does not automatically preserve its stability properties. Therefore, to inject context, we instead adapt the operator decomposition of~\cite{massai2026context} to the ICI predictor:
\begin{equation}\label{eq:ctx_split}
    \approxOp(\whatMat;\contextMat)=
    \approxOp_{\infty}(\whatMat,\contextMat)
    \boxtimes\approxOp_p(\whatMat).
\end{equation}
Here, $\approxOp_p\in\Lp^{SC}$ generates stable dynamical features from $\whatMat$, while the strictly causal mixer $\approxOp_\infty$ produces matrices $A_t$ satisfying $\sup_{t,\whatMat,\contextMat}\Vert A_t\Vert\leq\bar\sigma<\infty$. The product $\boxtimes$ acts pointwise: $(\mathbf A\boxtimes\mathbf z)_t=A_tz_t$. Intuitively, context changes how the features are scaled and combined, while their magnitude remains controlled by the stable module $\approxOp_p$.
Indeed, the pointwise bound $\Vert A_tz_t\Vert\leq\bar\sigma\Vert z_t\Vert$ gives~\cite[Thm.~1]{massai2026context}
\begin{equation}\label{eq:ctx_bound}
    \Vert\approxOp(\whatMat;\contextMat)\Vert_p
    \leq\bar\sigma\,\Vert\approxOp_p(\whatMat)\Vert_p,
    \qquad\forall\,\whatMat\in\lp,
\end{equation}
uniformly over context sequences. Thus, $\approxOp(\cdot\,;\contextMat)\in\Lp^{SC}$, and under the controller assumptions of~\ref{item:T1}, the resulting ICI model remains stabilized by $\cont$ for every admissible context.

Under the mild weighted-envelope input assumption and context-uniform envelope-preservation condition of~\cite[Thm.~2]{massai2026context}, the needed factorization~\eqref{eq:ctx_split} is also necessary and hence non-conservative. 

Finally, context can also be computed from predicted outputs, $\zeta_t=\kappa(\approxNoisyOut{t})$. The mixer's strict causality prevents an algebraic loop, and the same bound holds along the resulting trajectory. Consequently, known state-dependent context features can be supplied during simulation, as in Section~\ref{subsec:exp_context}, without requiring recorded context trajectories. In the sequel, for simplicity, we omit including context variables in the available data; indeed, they can always be incorporated in a plug-in fashion, as described so far, without altering the main result.

\section{Tractable Identification via the ICI Parameterization}
\label{sec:ERM_framework}

We now turn~\eqref{eq:ici_obj_theta} into a practical learning problem: we restrict it to finite horizons, introduce the direct (Section~\ref{sec:direct}) and indirect (Section~\ref{sec:indirect}) predictors, and fix the loss to the simulation error criterion (Section~\ref{sec:SEM}).

\subsection{Finite-Horizon Empirical Risk Minimization}
To address challenges~\ref{challenge:horizon} and~\ref{challenge:distribution}, we restrict~\eqref{eq:ici_obj_theta} to finite horizons and approximate the expectation empirically. Let $\mathcal{D}_T$ denote the distribution over closed-loop trajectories truncated to a finite horizon $T$, where a single trajectory is denoted by $\tau_T = (r_{0:T-1}, u_{0:T-1}, y_{0:T-1})$. We define the true finite-horizon identification risk as:
\begin{align}
    & J(\theta, \mathcal{D}_T) := \mathbb{E}_{\tau_T \sim \mathcal{D}_T} \left[ c \big( \theta, \tau_T \big) \right], \label{eq:true_loss} 
\end{align}
where $c(\cdot, \cdot)$ quantifies the discrepancy between the measured outputs and the predicted sequence $\hat{y}_{0:T-1}$. This is the quantity bounded in Section~\ref{sec:PAC}. Since $\realdataset_T$ is unknown, we rely on $\numSamples$ independent closed-loop experiments, $\mathbb{S}_T^N = \{ \tau_T^n \}_{n=0}^{N-1}$ with $\tau_T^n \sim \realdataset_T$ i.i.d., and approximate~\eqref{eq:true_loss} with the empirical risk:
\begin{align}
    & \hat{J}(\theta, \dataset_{T}^N) := \frac{1}{N} \sum_{n=0}^{N-1} \left[ c(\theta,\tau^n_{T}) \right]. 
    \label{eq:empirical_loss} 
\end{align}
The value of $\hat{J}$ depends on how $\hat{y}_{0:T-1}^n$ is generated from $\tau_T^n$: the following subsections introduce the direct and indirect strategies, which differ precisely in this choice.

\subsection{Direct ICI identification}\label{sec:direct}

In the direct ICI strategy the measured plant input is fed directly to the reconstructed model, $\{\hat{u}_{0:T-1}^n\}_{n=0}^{N-1} = \trueInData$, treating the closed-loop data as if collected in an open-loop experiment:
\begin{align}
    \min_{\theta \in \mathbb{R}^{d_\theta}} \quad & \hat{J}(\theta, \dataset_{T}^N) \label{eq:direct} \\
    \text{s.t.} \quad & \hat{y}_t^n = \hat{S}^{\theta}_t \bigl( u_{0:t-1}^n - K_{0:t-1} (\hat{y}_{0:t-1}^n) \bigr), \nonumber \\
    & \forall t=0,\dots,T-1, \; \forall n=0,\dots,N-1. \nonumber
\end{align}

The model $\approxSysICI^\theta$ built from $\approxOp^\theta$ through~\eqref{eq:ici} satisfies $\CL{\approxSysICI^\theta}{\cont} \in \Lp$, as guaranteed by point~\ref{item:T1} in Theorem~\ref{theo:ici}.

While direct methods lead to tractable optimization problems, they have two fundamental drawbacks.
\begin{enumerate}[label=\textcolor{darkgray}{\textbf{D.\arabic*}}]
    \item \label{item:D1} They can be used only when $\approxSys_{\text{ICI}}$ is stable~\cite{Forssell99Revisited},
    since with a simulation criterion, as adopted in Section~\ref{sec:SEM}, the risk $J$ is evaluated
    on open-loop predictions of $\approxSys_{\text{ICI}}$, which may diverge if the model is unstable.
     \item \label{item:D2} The regressor is noise-corrupted. Feedback makes $\noiseMat$ enter the trainable operator's input through $\trueInMat = \cont(\trueNoisyOutMat) + \excMat$ (\autoref{fig:true_CL}), so the regressor and the residual are correlated at first order. Even in the linear case this makes direct methods inconsistent under a misspecified noise model: the true $\trueSys$ is not recovered from infinite data~\cite{Soderstrom88SysID,Vandenhof98CLissues,Goodwin2002Bias}.
\end{enumerate}

\subsection{Indirect ICI identification}\label{sec:indirect}
\begin{figure}[htbp]
    \centering
    \resizebox{0.98\columnwidth}{!}{
    \resizebox{.9\linewidth}{!}{\begin{tikzpicture}[
        auto,
        block_pink/.style = {draw=black, rectangle, minimum height=0.6cm, minimum width=0.6cm, align=center, line width=0.8pt, fill=pink!40},
        block_white/.style = {draw=gray!90, rectangle, minimum height=0.6cm, minimum width=0.6cm, align=center, line width=1pt, fill=white},
        sum/.style = {draw=black, circle, minimum size=0.35cm, inner sep=0pt, line width=0.8pt, fill=white},
        line/.style = {draw, -{Latex[length=1.5mm, width=1.5mm]}, line width=0.8pt},
        font=\normalsize\bfseries
    ]

    \begin{scope}[shift={(0,0)}]
        \node[coordinate] (input_t) at (-0.8, 0) {};
        
        \node[sum] (sum1_t) at (0, 0) {};
        \node[sum] (sum2_t) at (1.0, 0) {};
        
        \node[block_pink] (Shat_t) at (2.0, 0) {$\mathbf{S}$};
        \node[block_pink] (K_inner_t) at (2.0, -0.9) {$\mathbf{K}$};
        
        \coordinate (tap_inner_t) at (3.0, 0);
        
        \node[sum] (sum3_t) at (3.5, 0) {};
        \node[coordinate] (v_t) at (3.5, 0.8) {};
        
        \coordinate (tap_outer_t) at (4.0, 0);
        \node[coordinate] (output_t) at (4.6, 0) {};
        
        \node[block_white] (K_outer_t) at (2.0, -1.9) {$\mathbf{K}$};

        \node[font=\tiny, scale=0.75] at ([xshift=-0.08cm]sum1_t.center) {$+$};
        \node[font=\tiny, scale=0.75] at ([yshift=-0.08cm]sum1_t.center) {$+$};
        \node[font=\tiny, scale=0.75] at ([xshift=-0.08cm]sum2_t.center) {$+$};
        \node[font=\tiny, scale=0.75] at ([yshift=-0.08cm]sum2_t.center) {$-$};
        \node[font=\tiny, scale=0.75] at ([xshift=-0.08cm]sum3_t.center) {$+$};
        \node[font=\tiny, scale=0.75] at ([yshift=0.08cm]sum3_t.center) {$+$};

        \path[line] (input_t) -- node[above] {$\mathbf{r}$} (sum1_t.west);
        \path[line] (sum1_t.east) -- node[above] {$\mathbf{u}$} (sum2_t.west);
        \path[line] (sum2_t.east) -- (Shat_t.west);
        
        \path[draw, line width=0.8pt] (Shat_t.east) -- (tap_inner_t);
        \path[line] (tap_inner_t) -- (sum3_t.west);
        \path[line] (tap_inner_t) |- (K_inner_t.east);
        \path[line] (K_inner_t.west) -| (sum2_t.south);
        
        \path[line] (v_t) -- (sum3_t.north);
        \node[above] at (v_t) {$\mathbf{v}$};
        
        \path[draw, line width=0.8pt] (sum3_t.east) -- (tap_outer_t);
        \path[line] (tap_outer_t) -- node[above] {$\mathbf{y}$} (output_t);
        
        \path[line] (tap_outer_t) |- (K_outer_t.east);
        \path[line] (K_outer_t.west) -| (sum1_t.south);

        \begin{scope}[shift={(0, -0.95)}]
            \node[font=\Large\bfseries] at (5.2, 0) {$\approx$};

            \node[coordinate] (r_in_lin_t) at (5.8, 0) {};
            \node[block_pink] (Shat_lin_t) at (7.0, 0) {$\mathbf{S}$};
            \node[sum] (sum_lin_t) at (8.2, 0) {};
            \node[coordinate] (v_in_lin_t) at (8.2, 1.0) {};
            \node[coordinate] (y_out_lin_t) at (9.2, 0) {};

            \node[font=\tiny, scale=0.75] at ([xshift=-0.08cm]sum_lin_t.center) {$+$};
            \node[font=\tiny, scale=0.75] at ([yshift=0.08cm]sum_lin_t.center) {$+$};

            \path[line] (r_in_lin_t) -- node[above] {$\mathbf{r}$} (Shat_lin_t.west);
            \path[line] (Shat_lin_t.east) -- (sum_lin_t.west);
            \path[line] (v_in_lin_t) -- (sum_lin_t.north);
            
            \node[above] at (v_in_lin_t) {$\mathbf{S}_L(\mathbf{K}_L(\mathbf{v})) + \mathbf{v}$};
            
            \path[line] (sum_lin_t.east) -- node[above] {$\mathbf{y}$} (y_out_lin_t);
        \end{scope}

        \begin{scope}[on background layer]
            \fill[gray!15, rounded corners=2mm] 
                ([xshift=-0.2cm, yshift=0.35cm]sum2_t.north west) 
                rectangle 
                ([xshift=0.2cm, yshift=-0.3cm]tap_inner_t |- K_inner_t.south);
                
            \draw[draw=gray!80!black, densely dotted, line width=1pt, rounded corners=4mm]
                (-1.1, 1.8) rectangle (9.4, -2.5);
        \end{scope}

        \node[font=\normalsize\bfseries, text=pink!60!red, anchor=south] at (2.0, 0.525) {$\mathbf{G}$}; 
        \node[fill=gray!80!black, draw=gray!60, line width=0.8pt, text=white, rounded corners=1.5mm, font=\normalsize, inner xsep=4pt, inner ysep=2pt] at (-0.1, 1.45) {Nonlinear};
    \end{scope}

    \begin{scope}[shift={(0, -5.0)}]
        \node[coordinate] (input_b) at (-0.8, 0) {};
        
        \node[sum] (sum1_b) at (0, 0) {};
        \node[sum] (sum2_b) at (1.0, 0) {};
        
        \node[block_pink] (Shat_b) at (2.0, 0) {$\mathbf{S}$};
        \node[block_pink] (K_inner_b) at (2.0, -0.9) {$\mathbf{K}$};
        
        \node[sum] (sum3_b) at (3.5, 0) {};
        \node[coordinate] (v_b) at (3.5, 1.0) {}; 
        
        \coordinate (tap_outer_b) at (4.0, 0);
        \node[coordinate] (output_b) at (4.6, 0) {};
        
        \node[block_white] (K_outer_b) at (2.0, -1.9) {$\mathbf{K}$};

        \node[font=\tiny, scale=0.75] at ([xshift=-0.08cm]sum1_b.center) {$+$};
        \node[font=\tiny, scale=0.75] at ([yshift=-0.08cm]sum1_b.center) {$+$};
        \node[font=\tiny, scale=0.75] at ([xshift=-0.08cm]sum2_b.center) {$+$};
        \node[font=\tiny, scale=0.75] at ([yshift=-0.08cm]sum2_b.center) {$-$};
        \node[font=\tiny, scale=0.75] at ([xshift=-0.08cm]sum3_b.center) {$+$};
        \node[font=\tiny, scale=0.75] at ([yshift=0.08cm]sum3_b.center) {$+$};

        \path[line] (input_b) -- node[above] {$\mathbf{r}$} (sum1_b.west);
        \path[line] (sum1_b.east) -- node[above] {$\mathbf{u}$} (sum2_b.west);
        \path[line] (sum2_b.east) -- (Shat_b.west);
        \path[line] (Shat_b.east) -- (sum3_b.west);
        
        \path[line] (v_b) -- (sum3_b.north);
        \node[above] at (v_b) {$(\mathbf{SK} + \mathbf{I})\mathbf{v}$};
        
        \path[draw, line width=0.8pt] (sum3_b.east) -- (tap_outer_b);
        \path[line] (tap_outer_b) -- node[above] {$\mathbf{y}$} (output_b);
        
        \path[draw, line width=0.8pt] (tap_outer_b) -- (tap_outer_b |- K_outer_b.east);
        \path[line] (tap_outer_b |- K_inner_b.east) -- (K_inner_b.east);
        \path[line] (tap_outer_b |- K_outer_b.east) -- (K_outer_b.east);
        \path[line] (K_inner_b.west) -| (sum2_b.south);
        \path[line] (K_outer_b.west) -| (sum1_b.south);

        \begin{scope}[shift={(0, -0.95)}]
            \node[font=\Large\bfseries] at (5.2, 0) {$\equiv$};

            \node[coordinate] (r_in_lin_b) at (5.8, 0) {};
            \node[block_pink] (Shat_lin_b) at (7.0, 0) {$\mathbf{S}$};
            \node[sum] (sum_lin_b) at (8.2, 0) {};
            \node[coordinate] (v_in_lin_b) at (8.2, 1.0) {};
            \node[coordinate] (y_out_lin_b) at (9.2, 0) {};

            \node[font=\tiny, scale=0.75] at ([xshift=-0.08cm]sum_lin_b.center) {$+$};
            \node[font=\tiny, scale=0.75] at ([yshift=0.08cm]sum_lin_b.center) {$+$};

            \path[line] (r_in_lin_b) -- node[above] {$\mathbf{r}$} (Shat_lin_b.west);
            \path[line] (Shat_lin_b.east) -- (sum_lin_b.west);
            
            \path[line] (v_in_lin_b) -- (sum_lin_b.north);
            \node[above] at (v_in_lin_b) {$(\mathbf{SK} + \mathbf{I})\mathbf{v}$};
            
            \path[line] (sum_lin_b.east) -- node[above] {$\mathbf{y}$} (y_out_lin_b);
        \end{scope}

        \begin{scope}[on background layer]
            \draw[draw=gray!80!black, densely dotted, line width=1pt, rounded corners=4mm]
                (-1.1, 1.8) rectangle (9.4, -2.5);
        \end{scope}

        \node[fill=gray!80!black, draw=gray!60, line width=0.8pt, text=white, rounded corners=1.5mm, font=\normalsize, inner xsep=4pt, inner ysep=2pt] at (-0.1, 1.45) {Linear};
    \end{scope}

    \end{tikzpicture}}
    }
    \caption{\textit{Top:} Closed-loop system connecting the true plant $\trueSys$ (represented within the ICI framework by the operator $\trueOp$) and the controller $\cont$ (left), alongside its linearized approximation (right).
    \textit{Bottom:} In the linear case, noise injection can be shifted to enable loop cancellation, yielding the exact simplified form on the right.}
    \label{fig:approx-cl-all}
\end{figure}
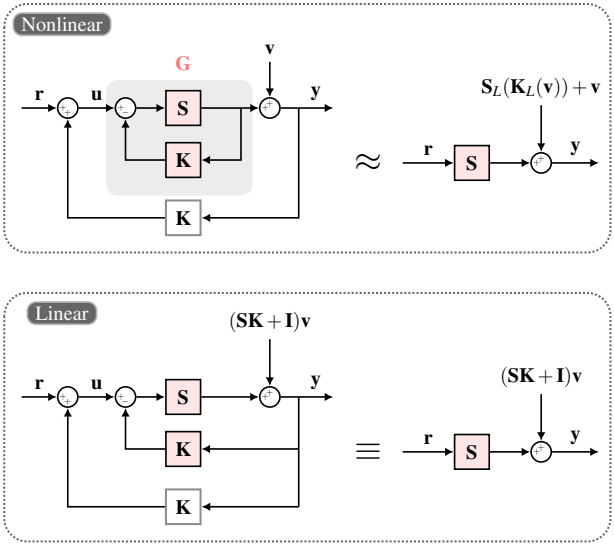
Indirect identification estimates the closed-loop response to the exogenous excitation $\excMat$ and uses the known controller to recover the plant~\cite{Soderstrom88SysID}. In the ICI framework, this amounts to learning the true operator $\trueOp \in \Lp^{SC}$ associated with $\trueSys \in \stabset$ by point~\ref{item:T2} in Theorem~\ref{theo:ici}, and recovering the plant through~\eqref{eq:ici}, evaluated by the recursion in Remark~\ref{remark:recursive}; this implements the dual-Youla parameterization idea~\cite{anderson1998youla,Desoer82GlobalPO} with the finite-dimensional stable operator models of Section~\ref{sec:unconstrained}. The closed loop under analysis is then equivalent to the interconnection of $\trueOp$ and $\cont$ in the top-left plot of~\autoref{fig:approx-cl-all}, and the measured output satisfies
\begin{align}
    \mathbf{y} = \trueOp ( \bm{\omega} ) + \noiseMat,
    \quad
    \bm{\omega} := \excMat + \cont(\mathbf{y}) - \cont(\mathbf{y} - \noiseMat) .
    \label{eq:strat3_y}
\end{align}
Identity~\eqref{eq:strat3_y} shows where the noise enters: additively at the output, and through the controller inside the argument of $\trueOp$; only the first entry is exogenous. Writing $\trueNoisyOutMat=\trueOp(\excMat)+\mathbf{e}$ gives the residual in closed form, $\mathbf{e}=\trueOp(\bm{\omega})-\trueOp(\excMat)+\noiseMat$, which vanishes with $\noiseMat$: then $\bm{\omega}=\excMat$ and $\trueOp$ reduces to the noise-free map $\excMat\mapsto\trueNoisyOutMat$ of $\CL{\trueSys}{\cont}$ in~\ref{item:T2}. Fitting $\trueOp$ from $(\excMat,\trueNoisyOutMat)$ thus has the form of an open-loop output-error problem, but with a non-exogenous disturbance whose dependence on $\excMat$ sets the bias. To expose this, assume $\trueSys_{0:T-1}$ and $\cont_{0:T-1}$ twice continuously differentiable near the operating trajectories, so that $\trueOp_{0:T-1}$ is as well by~\eqref{eq:lemma_recursion}, and linearize~\cite{anderson1998youla} $\trueOp$ around $\excMat$ and $\cont$ around $\trueOp(\excMat)$:
\begin{align}
    \mathbf{y}
    &\approx 
    \trueOp \bigl( \excMat \bigr) + \trueOp_L \bigl( \cont(\mathbf{y}) - \cont(\mathbf{y} - \approxNoiseMat) \bigr) + \approxNoiseMat
    \nonumber\\
    &\approx 
    \trueOp (\hspace{-24pt}\underbrace{\excMat}_{\textit{noise-independent signal}} \hspace{-24pt})
    + 
    \underbrace{\trueOp_L \bigl( \cont_L (\approxNoiseMat) \bigr) + \approxNoiseMat}_{\textit{noise}},
    \label{eq:linearized_out}
\end{align}
where $\trueOp_L := D\trueOp[\excMat]$ and $\cont_L := D\cont[\trueOp(\excMat)]$ are causal linear operators, block lower triangular by causality and strictly so for $\trueOp_L$. The second line follows by expanding both terms of $\cont(\trueNoisyOutMat)-\cont(\trueNoisyOutMat-\approxNoiseMat)$ at $\trueOp(\excMat)$: the constant terms cancel and the linear ones combine into $\cont_L(\approxNoiseMat)$. Since $\trueNoisyOutMat-\trueOp(\excMat)$ and $\trueNoisyOutMat-\approxNoiseMat-\trueOp(\excMat)$ are $O(\|\approxNoiseMat\|)$ by incremental $\ell_p$-stability of $\cont$ and local Lipschitz continuity of $\trueOp_{0:T-1}$, both approximations in~\eqref{eq:linearized_out} hold up to $O(\|\approxNoiseMat\|^2)$ with coefficients depending on $\excMat$ alone, so that, to first order in the noise, $\excMat$ propagates through $\trueOp$ alone (top-right plot of~\autoref{fig:approx-cl-all}). The same identity holds for the ICI model with $\approxOp^{\params}$ in place of $\trueOp$, so $\approxOp^{\params}(\excMat)$ is its first-order closed-loop response to $\excMat$, which we adopt as predictor. Accordingly, we adopt the open-loop learning problem
\begin{align}
    \min_{\theta \in \mathbb{R}^{d_\theta}} \quad & \hat{J}(\theta, \dataset_{T}^N) \label{eq:indirect} \\
    \text{s.t.} \quad & \hat{y}_t^n = \hat{S}^{\theta}_t \bigl(r_{0:t-1}^n\bigr), \nonumber \\
    & \forall t=0,\dots,T-1, \; \forall n=0,\dots,N-1, \nonumber
\end{align}
with $\approxSys_{\text{ICI}}$ reconstructed as the interconnection of the learned $\approxOp^{\theta}$ and the known $\cont$. Problem~\eqref{eq:indirect} is not an exact reformulation: its regressor is $\excMat$, while the signal driving $\trueOp$ in the experiment is $\bm{\omega}$, which equals $\excMat$ only for $\noiseMat=\zeroMat$. By~\eqref{eq:linearized_out} the first-order residual is linear in $\noiseMat$ with $\excMat$-dependent coefficients, hence of zero conditional mean given $\excMat$; the minimizer of~\eqref{eq:indirect} is therefore biased at $O(\Vert\noiseMat\Vert^{2})$, set by the curvature of $\trueOp$ and $\cont$, against the first-order correlation of~\ref{item:D2}. Superposition makes the separation exact in the linear case (Remark~\ref{rem:openloop}).

\begin{remark}[Validity of linear approximations]
\label{remark:validity_linear_approx}
The approximation~\eqref{eq:linearized_out} is exact in the absence of noise and when $\trueOp$ and $\cont$ are affine, with $\trueOp_L$ and $\cont_L$ their linear parts; in the linear case the equivalent output disturbance $(\trueOp\cont + \eyeOp)\noiseMat$ makes the loops cancel exactly (bottom plots of~\autoref{fig:approx-cl-all}). In the nonlinear case it is a local small-noise approximation: both $\noiseMat$ and the induced perturbation $\cont(\mathbf{y})-\cont(\mathbf{y}-\noiseMat)$ must stay where the Taylor remainders are small. Incremental $\ell_p$-stability of $\cont$ bounds that perturbation by $\gamma_\cont$ times the noise norm on each truncated horizon, and bounds the induced $\ell_p$-norm of every truncated Jacobian of $\cont$ by $\gamma_\cont$, so $\cont_L$ extends to a bounded operator on $\ell_p$. Accuracy therefore depends on local curvature and noise amplification; a high signal-to-noise ratio alone does not guarantee it.
\end{remark}

\begin{remark}[Properties of the indirect predictor]\label{rem:openloop}
    Fitting $\trueOp$ from the exogenous pair $(\excMat,\mathbf{y})$ via~\eqref{eq:indirect} avoids the open-loop simulation issue in~\ref{item:D1}: for every $\theta$, $\approxOp^\theta \in \Lp^{SC}$ implies $\hat{\mathbf{y}}\in\lp$ whenever $\excMat\in\lp$, even if $\trueSys$ is unstable.
    Regarding~\ref{item:D2}, the predictor is driven by $\excMat$, assumed independent of $\noiseMat$, rather than by noise-corrupted measurements; as discussed after~\eqref{eq:indirect}, this reduces the bias to second order but does not eliminate it.
    Consistency is thus limited to the linear case, where $\mathbf{y}=\trueOp(\excMat)+(\trueOp\cont+\eyeOp)\noiseMat$ exactly and $\trueOp\cont+\eyeOp=(\eyeOp-\trueSys\cont)^{-1}$ is the output sensitivity of the true loop: with zero-mean finite-variance noise independent of $\excMat$, output-error identification is consistent irrespective of the noise spectrum, provided the model class contains $\trueOp$, the excitation identifies it, and the usual statistical regularity and global optimization conditions hold~\cite{Forssell99Revisited}. Direct methods, instead, can require a correctly specified noise model to avoid feedback-induced bias~\cite{Vandenhof98CLissues,ljung1999system}.
\end{remark}

\subsection{The Simulation Error Objective}\label{sec:SEM}
A fundamental advantage of the unconstrained optimization problem~\eqref{eq:ici_obj_theta} is that closed-loop stability is guaranteed entirely by the ICI parameterization, so the identification loss $c$ can be freely designed to incorporate different modeling objectives without jeopardizing the stability of the learned model. For instance, the empirical loss could embed physics-informed penalties to enforce known physical constraints, state-dependent weightings to prioritize accuracy near critical equilibrium points or safety boundaries, or negative log-likelihoods to capture non-Gaussian noise profiles. In this work we adopt the standard Simulation Error Method (SEM): when the intended use of the model is simulation or control design, the multi-step simulation error is the natural fitting criterion and the standard choice for training deep networks on dynamical systems~\cite{Pillonetto25NN}, whereas one-step criteria optimize more easily but transfer poorly to free-run simulation. We define the trajectory discrepancy $c$ as the Mean Squared Error (MSE) between predicted and measured outputs, averaged over both the finite horizon $T$ and the $d_y$ output coordinates:
\begin{equation} \label{eq:mse_loss}
    c \big( \theta, \tau_T \big) := \frac{1}{T d_y} \sum_{t=0}^{T-1} \| \hat{y}_t - y_t \|_2^2,
\end{equation}
where the sequence of predicted outputs $\hat{\mathbf{y}}_{0:T-1}$ is generated according to either the direct or indirect identification strategy defined in the previous subsections.
Substituting \eqref{eq:mse_loss} into \eqref{eq:empirical_loss} yields an unconstrained least-squares problem for each strategy, efficiently minimized by gradient-based methods since $\approxOp^{\theta}$ is differentiable.

\section{PAC-Bayesian training for improved generalization}\label{sec:PAC}
System identification must operate on a finite set of noisy closed-loop trajectories. This finite-sample setting introduces a generalization gap: a model that successfully minimizes the empirical risk on the training data may still perform poorly on unseen realizations of excitation and noise. The goal of this section is to quantify this gap using PAC-Bayesian theory, and to translate the resulting bound into a principled, generalization-aware training objective for the ICI operator $\approxOp$.

\subsection{PAC-Bayes Generalization Bound}
To certify the finite-sample reliability of the identified model, we seek to formally upper-bound the true expected loss~\eqref{eq:true_loss} using its empirical estimate~\eqref{eq:empirical_loss}. We achieve this via the PAC-Bayesian framework, which casts the learning process as an update from a data-independent prior distribution $\P$ over the model parameters $\theta$, to a data-dependent posterior distribution $\Q$ based on the observed dataset. 

Within our system identification framework, both the prior $\P$ and the posterior $\Q$ are defined over the parameters, $\theta$, of the ICI operator $\approxOp^{\theta}$.
Crucially, this directly integrates the PAC-Bayesian framework with the ICI architecture: because every $\theta \in \mathbb{R}^{d_\theta}$ maps to an operator $\hat{S}^\theta \in \mathcal{L}^{SC}_p$, any model configuration sampled from the posterior is guaranteed by construction to be stabilized by the controller $\cont$.

PAC-Bayesian theory requires a bounded loss:
\begin{assumption}\label{assumption-id}
    The finite-horizon identification loss in~\eqref{eq:mse_loss} is upper-bounded by $C\in \R^+$,
    i.e.,
    $c(\theta,\tau_T) \in [0, C)$, for all $\theta \in \Theta := \mathbb{R}^{d_\theta}$ and $\tau_T$ in the support of $\realdataset_{T}$.
\end{assumption}

An unbounded loss is mapped to $[0, C)$ via:
\begin{align}
    \Tilde{c}(\theta,\tau_T) \coloneqq C \tanh(c(\theta,\tau_T)/\gamma), \label{eq:loss_trans}
\end{align}
with $\gamma \in \R^+$: the map is near-linear on $[0, \gamma]$ and saturates smoothly to $C$,
so a suitable $\gamma$ keeps training in the linear regime. Throughout this section $J$ and
$\hat J$ are understood on the transformed loss $\Tilde c$.

We now bound the true identification risk $J$ in~\eqref{eq:true_loss}.
\begin{theorem}[Adapted from Theorem~2.7 in~\cite{userfriendly}]
\label{theo:pac-id}
    Consider the distribution of closed-loop trajectories $\realdataset_{T}$, a dataset $\mathbb{S}_T^N = \{ \tau_T^n \}_{n=0}^{N-1} = \bigl\{( \exc{0:\horizon-1}^\counter,u_{0:\horizon-1}^\counter,\trueNoisyOut{0:\horizon-1}^\counter ) \bigr\}_{\counter=0}^{\numSamples-1} \sim \realdataset_{T}$, a prior distribution $\P$ independent of $\mathbb{S}_T^N$, and the identification loss $c$.
    Under Assumption~\ref{assumption-id}, for every $\lambda>0$, confidence level $\delta \in (0,1)$, posterior distribution $\Q$, and model parameters $\theta \sim \Q$, the inequality
    \begin{align}
        J(\theta, \realdataset_{T})
        \leq \hat{J}(\theta, \mathbb{S}_T^N)
        + \frac{1}{\lambda}\Bigl(
            \ln \frac{d\Q}{d\P}(\theta) + \ln \frac{1}{\delta}
        \Bigr)
        + \frac{\lambda C^2}{8 N} \label{eq:bound-id}
    \end{align}
    holds with probability at least $1-\delta$ over simultaneously sampling the $N$ i.i.d. trajectories of $\dataset_{T}^N$ from $\realdataset_{T}$ and $\theta \sim \Q$. The true and empirical risks, $J$ and $\hat{J}$, are defined in~\eqref{eq:true_loss} and~\eqref{eq:empirical_loss} respectively.
\end{theorem}
The bound combines the empirical risk, a log-density ratio measuring the deviation of $\Q$ from $\P$ at $\theta$, and a $\theta$-independent complexity term. It adapts Theorem~2.7 in~\cite{userfriendly} to time-series data by treating independent closed-loop trajectories as the i.i.d.\ samples, so that $c$ compares entire predicted and observed sequences.
Standard PAC-Bayesian bounds control the risk averaged over $\Q$, i.e., the risk of a randomized predictor that redraws $\theta \sim \Q$ at every use, which is at odds with deploying one identified model. Theorem~\ref{theo:pac-id} is instead a disintegrated bound: it holds for a single parameter drawn once from the posterior, $\theta \sim \Q$, matching the identification objective of deploying one model.

\subsection{Optimal Posterior}\label{subsec:qstar-id}
Unlike Bayesian inference, $\Q$ need not come from Bayes' theorem: we select it to minimize the expected value, over $\theta \sim \Q$, of the upper bound in~\eqref{eq:bound-id}:
\begin{align}
    \Q^* \coloneqq& \argmin_{\Q} \E_{\theta \sim \Q}  \Bigl[\hat{J}(\theta, \dataset_{T}^N) 
    + \frac{1}{\lambda} \Bigl( \ln \frac{d\Q}{d\P}(\theta)
    + \ln \frac{1}{\delta} \Bigr) 
    + \frac{\lambda C^2}{8 N} \Bigr] \nonumber \\
    =& \argmin_{\Q} \E_{\theta \sim \Q}  \Bigl[\hat{J}(\theta, \dataset_{T}^N) 
    + \frac{1}{\lambda} \ln \frac{d\Q}{d\P}(\theta)\Bigr] \nonumber\\
    =& \argmin_{\Q} \E_{\theta \sim \Q}  \Bigl[\hat{J}(\theta, \dataset_{T}^N) \Bigr] +
    \frac{1}{\lambda} \KL\bigl(\Q \Vert \P \bigr) 
    ,\label{eq:def_qstar}
\end{align}
where $\KL$ stands for the Kullback–Leibler divergence.
The objective trades off, through $\lambda$, empirical risk against closeness to the prior (which can encode structural knowledge or regulate complexity). In the two limiting regimes, $\Q^*$ concentrates on the minimizers of $\hat{J}$ as
$\lambda \to \infty$, and $\Q^* \to \P$ as $\lambda \to 0$.

\begin{corollary}[Lemma~1.1.3 in~\cite{catoni2007pac}]\label{corol:qstar-id}
    The minimizer of~\eqref{eq:def_qstar} is the \textit{Gibbs} distribution, characterized by:
    \begin{align}
    \Q^*(\theta)
    &= \P(\theta) \,
    e^{-\lambda \hat{J}(\theta, \dataset_{T}^N)}
    / Z_\lambda(\P,\dataset_{T}^N) \label{eq:qstar-id},\\
    Z_\lambda(\P,\dataset_{T}^N) &\coloneqq \E_{\theta \sim \P}
    e^{-\lambda \hat{J}(\theta, \dataset_{T}^N)}, \label{eq:Z-id}
    \end{align}
    where $Z_\lambda(\P,\dataset_{T}^N)$ normalizes the distribution $\Q^*$.
\end{corollary}
Substituting the optimal posterior~\eqref{eq:qstar-id} into~\eqref{eq:bound-id} yields the tightest bound:
\begin{align}
    J(\theta, \realdataset_{T})
    \leq
    \frac{1}{\lambda} \Bigl( \ln \frac{1}{\delta}
    - \ln Z_\lambda(\P,\dataset_{T}^N) \Bigr)
    + \frac{\lambda C^2}{8 N} ,
    \label{eq:bound_qstar-id}
\end{align}
holding with probability at least $1-\delta$ over simultaneously sampling $\dataset_{T}^N$ and $\theta \sim \Q^*$.
The term $- 1/\lambda \ln Z_\lambda(\P,\dataset_{T}^N)$ lies in $[0, C)$ and decreases when the (data-independent) prior $\P$ assigns more mass to parameters with lower empirical risk~\eqref{eq:Z-id}: the bound tightens when $\P$ is ``aligned'' with $\dataset_{T}^N$. The bound~\eqref{eq:bound_qstar-id} holds for the optimal posterior $\Q^*$ and does not depend on the drawn $\theta$. Since $\Q^*$ minimizes the expectation over $\theta \sim \Q$ of the right-hand side of~\eqref{eq:bound-id}, it yields the tightest instance of~\eqref{eq:bound-id} in expectation over the posterior draw, though not necessarily for every realization of $\theta$.

Tightening~\eqref{eq:bound_qstar-id} amounts to tuning $\lambda>0$, which Theorem~\ref{theo:pac-id} leaves free. Since $-1/\lambda \ln Z_\lambda(\P, \dataset_{T}^N)$ requires an expectation~\eqref{eq:Z-id} and lies in $[0, C)$ under Assumption~\ref{assumption-id}, we replace it by $C$ and select $\lambda$ by minimizing the resulting looser bound. 

\begin{proposition}\label{prop:lambda}
    The optimal $\lambda>0$ that minimizes the upper bound on the right-hand side of~\eqref{eq:bound_qstar-id} obtained by replacing $-1/\lambda \ln Z_\lambda(\P, \dataset_{T}^N)$ with $C$, is given by:
    \begin{align}
        \lambda^* = \frac{\sqrt{8 N \ln (1/\delta)}}{C}. \label{eq:lambdastar-id}
    \end{align}
\end{proposition}

The same $\lambda^*$ is obtained by replacing $-1/\lambda \ln Z_\lambda(\P, \dataset_{T}^N)$ with the tighter bound $\E_{\theta \sim \P}\hat{J}(\theta, \dataset_{T}^N)$, which follows from Jensen's inequality and is likewise independent of $\lambda$. Crucially, $\lambda^*$ does not depend on $\dataset_{T}^N$, as required for Theorem~\ref{theo:pac-id} to hold with $\lambda = \lambda^*$.

Given $\P$ and $\dataset_{T}^N$,~\eqref{eq:qstar-id} and~\eqref{eq:lambdastar-id} fully specify the optimal posterior. In practice, hyperparameters of the prior (e.g., its width) and $\lambda$ may also be selected from a finite, data-independent set of $m$ candidates by minimizing the bound: by a union bound over the candidates,~\eqref{eq:bound_qstar-id} holds for the selected pair with probability at least $1-\delta$ upon replacing $\delta$ with $\delta/m$, including in~\eqref{eq:lambdastar-id}.

\subsection{Sampling from the posterior}
\label{subsec:numerical_SVGD}
The optimal posterior $\Q^*$ in~\eqref{eq:qstar-id} is known only up to the normalization
constant $Z_\lambda(\P, \dataset_{T}^N)$, and sampling from unnormalized distributions is
generally difficult: Markov chain Monte Carlo methods are applicable in principle, but are
typically expensive for large datasets~\cite{green2015bayesian}.
Following~\cite{Alquier2016variational}, we adopt instead a variational approximation, in which
$\Q^*$ is replaced by a tractable distribution that is easy to sample from. Specifically, we
employ Stein Variational Gradient Descent (SVGD)~\cite{SVGD}, which represents the posterior by
$k \in \mathbb{N}$ particles $\{\phi_1, \dots, \phi_k\}$. Each particle $\phi_\kappa$,
$\kappa = 1, \dots, k$, collects the parameters of one model, i.e., $\phi_\kappa \in \thetaset$;
the particles are initialized from the prior and iteratively updated by gradient descent, so
that their empirical density in parameter space approximates $\Q^*$. Since the updates are
gradient-based, SVGD integrates naturally with our unconstrained parameterization. After
training, $\Q^*$ is approximated by $\hat{\Q}_k$, the uniform distribution over the resulting
particles, so that sampling from the approximate posterior amounts to selecting one particle
uniformly at random. The use of SVGD for approximating Gibbs posteriors arising in PAC-Bayesian methods
is well established, e.g.~\cite{PACOH,PACPFL} and we refer the reader to~\cite{SVGD} for further implementation details.
One caveat is worth noticing: Theorem~\ref{theo:pac-id} and its
specialization~\eqref{eq:bound_qstar-id} certify a \emph{single} draw from $\Q^*$, whereas SVGD
provides the particle approximation $\hat{\Q}_k$. The value of~\eqref{eq:bound_qstar-id} depends only on $\P$ and $\dataset_{T}^N$, not on the particles, and~\eqref{eq:bound-id} cannot be applied to $\hat{\Q}_k$ directly, since a discrete particle distribution is not absolutely continuous with respect to $\P$. The certificate therefore refers to $\Q^*$; the deployed particle is an approximate draw from it, and the discrepancy between $\Q^*$ and $\hat{\Q}_k$ is not quantified by the bound.

\section{Numerical experiments}
\label{sec:experiments}

We evaluate the ICI framework across four scenarios validating: (i) the identification of open-loop unstable plants with guaranteed closed-loop stability, (ii) accuracy against standard baselines on stable systems, (iii) context-aware identification across environment families, and (iv) the generalization benefits of PAC-Bayesian training under data scarcity.
Section~\ref{subsec:setup} describes the benchmark systems, the compared strategies, the evaluation metric, and the architectures, sample sizes, and signal levels of each experiment; layer widths, learning rates, whitening, training safeguards, per-seed protocols, and the complete result grids are in the public repository at \url{https://github.com/DecodEPFL/Nonlinear_system_identification_PAC.git}.

\subsection{Experimental Setup}
\label{subsec:setup}

\textbf{Target Systems.}
Throughout the experiments, we benchmark our methods on two distinct dynamical systems:
\begin{enumerate}
    \item A nonlinear, open-loop unstable system defined by the state-space equations:
    \begin{equation}
    \begin{cases}
    \trueState{t+1} = 3\trueState{t} + \tanh(\trueState{t}) + \trueIn{t}, \quad \trueState{0} = \bar{\trueState{}} \\
    \trueNoisyOut{t} = \trueState{t} + \noise{t}
    \end{cases}
    \label{eq:scalar_sys}
    \end{equation}
Being open-loop unstable, the system is operated in closed loop during data collection,
stabilized by the nonlinear controller
$
    \contFunc{t}(\trueNoisyOut{0:t}) = -3.5\,\trueNoisyOut{t}+0.3\tanh(\trueNoisyOut{t}).
$
This controller satisfies the hypotheses of Theorem~\ref{theo:ici} globally:
$\cont(\zeroMat)=\zeroMat$, and
$|\contFunc{}'(\trueNoisyOut{})| = |-3.5+0.3\sech^2(\trueNoisyOut{})| \leq 3.5$ for every
$\trueNoisyOut{}$, so $\cont$ is incrementally $\ell_p$-stable with incremental gain
$\gamma_\cont = 3.5$.
Together with $\cont(\zeroMat) \in \lp$ this yields $\cont \in \Lp^C$.
The system input is then $\trueIn{t} = \contFunc{t}(\trueNoisyOut{0:t}) + \exc{t}$.

\item A stable planar robot tracking a target, modeled as a point mass $m \in \mathbb{R}_+$~\cite{PerfBoost}. The state comprises position $\trueState{1,t} \in \R^2$ and velocity $\trueState{2,t} \in \R^2$, the output is a noisy position measurement, and the input is a planar force:
\begin{align*}
\begin{cases}
    \begin{bmatrix}
        \trueState{1,t+1} \\ \trueState{2,t+1}
    \end{bmatrix}
    =
    \begin{bmatrix}
        \trueState{1,t} \\ \trueState{2,t}
    \end{bmatrix}
    + T_s
    \begin{bmatrix}
        \trueState{2,t} \\
        \frac{1}{m} \bigl(-C(x_{1,t}, x_{2,t}) + \trueIn{t} \bigr)
    \end{bmatrix}, \\
    \trueNoisyOut{t} = \trueState{1,t} + \noise{t},
\end{cases}
\end{align*}

where $T_s > 0$ is the sampling time, and the function $C(x_{1,t}, x_{2,t}) = b_1(x_{1,t}) x_{2,t} + b_2 \|x_{2,t}\|_2 x_{2,t}$, with $0 < b_2 < b_1(x_{1,t}) \;\forall x_{1,t} \in \R^2$, models nonlinear drag forces, such as air or water resistance. The coefficient $b_1(x_{1,t})$ can represent either a uniform environment or a position-dependent spatial friction map.
The robot is subject to nonlinear drag. To regulate its position toward the target $\trueNoisyOut{}^* = \zeroMat \in \R^2$, we apply a static proportional controller with gain matrix $K = \operatorname{diag}(\kappa_1, \kappa_2)$, where $\kappa_1, \kappa_2 > 0$. The control law is $\contFunc{t}(\trueNoisyOut{0:t}) = K(\trueNoisyOut{}^* - \trueNoisyOut{t})$, and the system input is $\trueIn{t} = \contFunc{t}(\trueNoisyOut{0:t}) + \exc{t}$.
\end{enumerate}

Unless otherwise specified, the reference excitation $r_t$ is independently sampled at each time step from a zero-mean Gaussian distribution, $\exc{t} \sim \mathcal{N}(0, \sigma_r^2 \eye)$, and the output is corrupted by measurement noise $v_t\sim\mathcal N(0,\sigma_v^2\eye)$; both are $\lp$ on the finite horizon $\horizon$ of each trajectory, so Theorem~\ref{theo:ici} applies.
While persistent excitation is most commonly defined and analyzed in linear settings, in this nonlinear setting Gaussian excitation qualitatively promotes exploration of an informative range of the state space~\cite{ljung1999system}.

\textbf{Baselines.}
To highlight the benefits of the ICI stability certificate, we benchmark our approaches (both the indirect strategy~\eqref{eq:indirect} and the direct counterpart~\eqref{eq:direct}) against two uncertified baselines that minimize the simulation error~\eqref{eq:mse_loss}.
Instead of incorporating structural guarantees, these baselines parameterize $\approxSys$ as a generic neural network, discarding the prior knowledge that the true plant is stabilized by $\cont$.
The two baseline strategies differ fundamentally in how the networks are simulated and, consequently, in which data drive the predictor:

\begin{itemize}
    \item \textbf{Direct OE} (Direct Output Error).
This is the standard direct identification baseline,
widely used in neural system
identification~\cite{Pillonetto25NN,ribeiro2020smoothness}.
A generic neural model $\approxSys_\theta$ is trained
to minimize the simulation error on the
recorded input $\mathbf{u}$:
\begin{align}
    \min_{\theta \in \mathbb{R}^{d_\theta}} \;\;
    &\frac{1}{N} \sum_{n=0}^{N-1} \left[
        \frac{1}{T} \sum_{t=0}^{T-1}
        \frac{1}{d_y}\|\hat{y}_t^n - y_t^n\|_2^2
    \right]
    \label{eq:directOE} \\
    \text{s.t.} \;\;
    &\hat{y}_t^n
       = \hat{G}_t^\theta\bigl(u_{0:t-1}^n\bigr),
    \nonumber \\
    &\forall\, t{=}0,\dots,T{-}1,
    \;\; \forall\, n{=}0,\dots,N{-}1.
    \nonumber
\end{align}
\item \textbf{CLOE} (Closed-Loop Output Error).
This strategy~\cite{landau1997cloe,landau2001recursive} simulates the closed-loop system driven by the exogenous reference $\mathbf{r}$. By placing the neural model $\approxSys_\theta$ in feedback with the known controller $\mathbf K$, thus closing the loop around the model's own predictions, CLOE avoids feeding the noisy recorded input into the predictor:
        \begin{align}
            \min_{\theta \in \mathbb{R}^{d_\theta}}\;\; &\frac{1}{N} \sum_{n=0}^{N-1} \left[ \frac{1}{T} \sum_{t=0}^{T-1} \frac{1}{d_y}\| \hat{y}_t^n - y_t^n \|_2^2\right] \label{eq:baseline}
            \\
            \text{s.t.} \; &\hat{y}_t^n = \hat{G}_{t}^\theta \bigl( r_{0:t-1}^n+ K_{0:t-1} (\hat{y}_{0:t-1}^n) \bigr) \nonumber, \\
            & \forall t=0,\dots,T-1, \; \forall n=0,\dots,N-1. \nonumber
        \end{align}
\end{itemize}

\textbf{Evaluation metrics.}
All identification strategies are evaluated using the FIT metric on a fixed, noise-free test set, common to every method within each training condition, in both open-loop (OL) and closed-loop (CL) settings; measurement noise affects the training data only, so the test isolates the quality of the identified dynamics.

Goodness of fit is computed per test trajectory $\counter$ as $\mathrm{FIT}^{\counter}=100\bigl(1-\Vert \trueNoisyOut{0:\horizon_{test}-1}^{\counter}-\hat{y}^{\counter}_{0:\horizon_{test}-1}\Vert_2/\Vert \trueNoisyOut{0:\horizon_{test}-1}^{\counter}-\bar{\trueNoisyOut{}}^{\counter}\Vert_2\bigr)$ [\%], where $\trueNoisyOut{t}^{\counter}$ is the noise-free true output, $\hat{y}_{t}^{\counter}$ the model's, and $\bar{\trueNoisyOut{}}^{\counter}$ the temporal mean of trajectory $\counter$, subtracted at every $t$; FIT ranges in $(-\infty,100]$~\cite{Pillonetto25NN}. For open-loop-unstable plants, the open-loop comparison is shown graphically (on a logarithmic scale); closed-loop FIT/RMSE are computed over the bounded rollouts only, with the divergence rate reported separately.

We report the mean $\pm$ one standard deviation \emph{across the $N_{test}$ test trajectories}; training seeds are fixed and all protocols are deterministic given the seed, so every number is exactly reproducible from the public repository.

\textbf{Models and datasets.}
\begin{table*}[t]
\centering
\caption{Model classes and data of the four experiments. $N$, $T$ are the number and length
of the training trajectories, $\sigma_r$, $\sigma_v$ the excitation and measurement-noise
levels, and $N_{test}$, $\horizon_{test}$ the number and length of the test trajectories.
Initial conditions are $\bar{\trueState{}}\sim\mathcal U(-2,2)$ (training) and
$\bar{\trueState{}}\sim\mathcal U(-4,4)$ (redeployment test) in
Section~\ref{subsec:exp_unstable}, and $p_0\sim\mathcal U([-10,10]^2)$ with zero initial
velocity elsewhere. The redeployment test of Section~\ref{subsec:exp_unstable} is repeated at
$\sigma_v \in \{0,\,0.3,\,0.6\}$ for each of the $S=40$ training runs per strategy.}
\label{tab:app_details}
\footnotesize
\setlength{\tabcolsep}{5pt}
\resizebox{\textwidth}{!}{%
\begin{tabular}{@{}lllllrrllrr@{}}
\toprule
Section & System & Observation & ICI operator $\approxOp$ & Baseline class & $N$ & $T$ & $\sigma_r$ & $\sigma_v$ & $N_{test}$ & $\horizon_{test}$\\
\midrule
\ref{subsec:exp_unstable} & 1 & full state & deep SSM & ReLU RNN & $40$ & $100$ & $0.5$ & $0.1$ & $100$ & $200$\\
\ref{subsec:exp_robot} & 2 & positions & deep SSM & LSTM & $200$ & $100$ & $\{0.5,\dots,20\}$ & $\{0,\dots,2\}$, AR(1) & $100$ & $100$\\
\ref{subsec:exp_context} & 2 & full state & deep SSM  & deep SSM, no context & $60$ & $50$ & $20$ & $0.1$ & $100$/$20$ & $50$\\
\ref{subsec:exp_pacbayes} & 2 & positions & REN  & --- & $1$--$64$ & $50$ & $20$ & $0.5$ & $100$ & $50$\\
\bottomrule
\end{tabular}}
\end{table*}

Table~\ref{tab:app_details} collects the model classes, sample sizes, and signal levels of the four experiments.

\subsection{Feedback-Stabilized Unstable Scalar System}
\label{subsec:exp_unstable}
This experiment validates on the unstable plant~\eqref{eq:scalar_sys} that
(a)~every identified ICI model $\approxSysICI$ is guaranteed to be stabilized by the data-collection controller, whereas a model returned by an uncertified method may not be, and (b)~using the indirect approach training does not diverge and recovers the open-loop unstable dynamics. The two ICI strategies use a contractive SSM as the operator $\approxOp$, while the two baselines use a vanilla RNN with unbounded ReLU activations, a class chosen because saturating networks can't represent unstable dynamics and would fail the second question by construction. The baselines are given more capacity than the ICI models, so that no conclusion can be attributed to a capacity advantage of the latter.

\textbf{Results.}
Table~\ref{tab:unstable} summarizes the structural properties of each method, while Figure~\ref{fig:unstable_ol_bands} shows the open-loop behavior. The analysis yields two main observations:

\emph{(a) Only the ICI framework guarantees closed-loop stability.}
No ICI model destabilized the loop in any of the $S=40$ runs, at any noise level, as
Theorem~\ref{theo:ici} guarantees. 
Among the uncertified methods, Direct OE returned a model that $\cont$ fails to stabilize in $25\%$ of the runs, with the loop diverging from a large fraction of the initial conditions. CLOE with the ReLU-RNN instead generates an unstable closed loop in $2.5\%$ of the runs.

\emph{(b) Indirect ICI ensures well-posed training and identifies unstable dynamics.}
As shown in Figure~\ref{fig:unstable_ol_bands}, the true plant exhibits unstable open-loop trajectories, which are accurately reproduced by indirect ICI. Both
direct strategies fail in every run: as anticipated by Drawback~\ref{item:D1}, they simulate the
candidate model on the recorded input, and for unstable models the finite-cost region shrinks
exponentially around the true plant~\cite{ribeiro2020smoothness}, so unconstrained gradient
descent diverges. CLOE avoids open-loop simulation, but its rollout stays bounded only while gradient updates keep the model stabilized by $\cont$, which nothing guarantees. For a fair comparison, we therefore equip these methods with heuristic training safeguards (horizon curricula, gradient clipping, learning-rate backoff~\cite{ribeiro2020smoothness}); with them, the direct approaches converge, but to stable surrogates that miss the open-loop instability. In contrast, indirect ICI restricts the search space to models that $\mathbf{K}$ certifiably stabilizes, keeping every training rollout bounded and recovering the
unstable dynamics without such heuristics.

\begin{table}[tb]
\centering
\caption{Unstable scalar plant~\eqref{eq:scalar_sys}: identifiability of the
unstable plant, certified closed-loop stability with the known controller,
and fraction of the $S=40$ independent training runs whose model the
data-collection controller fails to stabilize ($\sigma_{\noise{}}=0.3$).}
\label{tab:unstable}
\footnotesize
\begin{tabular}{l c c c}
\toprule
Strategy & \shortstack{Identifiability\\of unstable \textbf{G}} & \shortstack{Certified \\CL stability} & \shortstack{Unstable\\CL runs}\\
\midrule
Indirect ICI & \checkmark  & \textbf{\checkmark } & $0\%$\\
Direct ICI  & --- & \textbf{\checkmark } & $0\%$\\
CLOE & \checkmark & --- & $2.5\%$\\
Direct OE & --- & --- & $25\%$\\
\bottomrule
\end{tabular}
\end{table}

\begin{figure}[t]
\centering
\includegraphics[width=.8\linewidth]{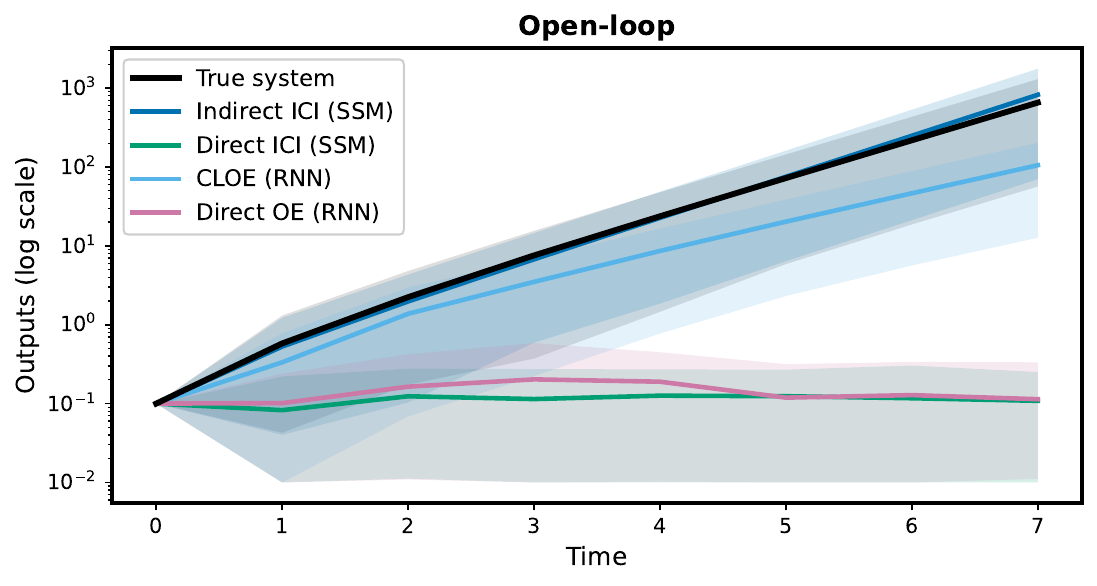}
\caption{Unstable scalar plant. Open-loop trajectories (log scale, $8$ steps) from $\bar{\trueState{}}=0.1$ of the models returned by $S=40$ independent training runs
(solid $=$ median over runs, shaded $=$ $10$--$90$ percentile band); the indirect strategies reproduce the diverging behavior, while both direct strategies return an inert surrogate.}
\label{fig:unstable_ol_bands}
\end{figure}

\subsection{Point-Mass Robot under Partial Observation}
\label{subsec:exp_robot}
This experiment evaluates the accuracy of the ICI strategies against the
generic-network baselines on the stable point-mass robot of Section~\ref{subsec:setup} in a uniform
environment ($b_1(x_{1,t}) \equiv b_1$), under partial observation: only the (noisy)
positions are measured, the velocities remain hidden.
We evaluate the closed-loop identification performance through two distinct parameter sweeps: varying the reference excitation level $\sigma_r$ on noise-free data, and varying the level of colored measurement noise $\sigma_v$ under constant excitation.
The two ICI strategies share the same certified deep SSM as the operator $\approxOp$ in~\eqref{eq:direct} and~\eqref{eq:indirect}, and the two baselines share a generic LSTM as the model $\approxSys_\theta$ in~\eqref{eq:directOE} and~\eqref{eq:baseline}.

\begin{table*}[!t]
\caption{Point-mass robot under partial observation (Section~\ref{subsec:exp_robot}): open-loop and closed-loop FIT [\%], mean $\pm$ std over the $100$ test trajectories. \emph{Top:} nominal excitation $\sigma_r{=}10$ and AR(1) measurement noise ($\rho{=}0.9$) of increasing level $\sigma_v$. \emph{Bottom:} noise-free training ($\sigma_v{=}0$) and increasing excitation level $\sigma_r$; at $\sigma_r{=}0.5$ the LSTM baselines fall below $-400\%$ open-loop FIT. Colored cells indicate the best result per column (ties are highlighted in both entries). Both sweeps are shown in Figure~\ref{fig:robot_sweeps}; the complete grids, including $\sigma_r \in \{0.5, 20\}$, are in the repository.}
\label{tab:robot_summary}
\vspace{1mm}
\begin{center}
\resizebox{0.9\textwidth}{!}{
\begin{tabular}{lcccccccc}
\toprule
& \multicolumn{2}{c}{$\sigma_v{=}0.1$} & \multicolumn{2}{c}{$\sigma_v{=}0.5$} & \multicolumn{2}{c}{$\sigma_v{=}1$} & \multicolumn{2}{c}{$\sigma_v{=}2$}\\
\cmidrule(lr){2-3}\cmidrule(lr){4-5}\cmidrule(lr){6-7}\cmidrule(lr){8-9}
\textbf{Strategy} & OL & CL & OL & CL & OL & CL & OL & CL\\
\midrule
Direct OE (LSTM)   & $72.3 \pm 13.3$ & $83.0 \pm 7.2$ & $66.3 \pm 17.2$ & $80.7 \pm 7.9$ & $45.9 \pm 22.5$ & $72.4 \pm 11.9$ & $-36.2 \pm 55.0$ & $43.5 \pm 30.7$\\
CLOE (LSTM)        & $78.9 \pm 10.6$ & $87.0 \pm 5.4$ & $70.2 \pm 16.5$ & $82.8 \pm 8.4$ & $39.1 \pm 34.1$ & $77.1 \pm 14.5$ & $-86.1 \pm 123.0$ & $30.3 \pm 54.8$\\
Direct ICI (SSM)   & $82.9 \pm 6.5$ & \cellcolor{\bestres}$92.1 \pm 3.1$ & $79.0 \pm 6.8$ & \cellcolor{\bestres}$91.7 \pm 3.0$ & $62.0 \pm 9.7$ & \cellcolor{\bestres}$90.2 \pm 3.0$ & $-7.0 \pm 33.6$ & \cellcolor{\bestres}$87.8 \pm 4.3$\\
Indirect ICI (SSM) & \cellcolor{\bestres}$86.7 \pm 5.1$ & \cellcolor{\bestres}$92.0 \pm 4.4$ & \cellcolor{\bestres}$84.3 \pm 6.0$ & \cellcolor{\bestres}$91.7 \pm 3.7$ & \cellcolor{\bestres}$74.0 \pm 10.0$ & $89.6 \pm 3.9$ & \cellcolor{\bestres}$60.9 \pm 15.7$ & $82.7 \pm 6.1$\\
\midrule
& \multicolumn{2}{c}{$\sigma_r{=}1$} & \multicolumn{2}{c}{$\sigma_r{=}2$} & \multicolumn{2}{c}{$\sigma_r{=}5$} & \multicolumn{2}{c}{$\sigma_r{=}10$}\\
\cmidrule(lr){2-3}\cmidrule(lr){4-5}\cmidrule(lr){6-7}\cmidrule(lr){8-9}
\textbf{Strategy} & OL & CL & OL & CL & OL & CL & OL & CL\\
\midrule
Direct OE (LSTM)   & $-170.1 \pm 118.6$ & $86.0 \pm 8.9$ & $-64.9 \pm 71.3$ & $83.3 \pm 10.1$ & $45.1 \pm 32.1$ & $86.1 \pm 6.0$ & $76.0 \pm 12.3$ & $85.1 \pm 4.9$\\
CLOE (LSTM)        & $-136.4 \pm 135.5$ & $89.7 \pm 4.1$ & $-59.5 \pm 107.3$ & $89.7 \pm 5.9$ & $69.3 \pm 20.1$ & $91.3 \pm 3.3$ & $82.3 \pm 10.1$ & $90.1 \pm 4.8$\\
Direct ICI (SSM)   & $-27.6 \pm 28.2$ & $95.3 \pm 3.6$ & $68.6 \pm 10.6$ & $95.4 \pm 3.9$ & $78.9 \pm 7.0$ & \cellcolor{\bestres}$95.6 \pm 2.5$ & $83.2 \pm 6.7$ & \cellcolor{\bestres}$92.3 \pm 3.0$\\
Indirect ICI (SSM) & \cellcolor{\bestres}$64.4 \pm 15.6$ & \cellcolor{\bestres}$96.5 \pm 2.6$ & \cellcolor{\bestres}$78.7 \pm 9.2$ & \cellcolor{\bestres}$96.1 \pm 2.6$ & \cellcolor{\bestres}$86.3 \pm 5.6$ & \cellcolor{\bestres}$95.5 \pm 2.3$ & \cellcolor{\bestres}$87.0 \pm 5.0$ & \cellcolor{\bestres}$92.0 \pm 4.4$\\
\bottomrule
\end{tabular}
}
\end{center}
\end{table*}

\textbf{Results.} Figure~\ref{fig:robot_sweeps} visualizes both sweeps, with their corresponding FIT measures detailed in Table~\ref{tab:robot_summary}. The analysis yields three main observations.

\emph{(a) Informativity at low excitation.} At low excitation (e.g.,
$\sigma_r=1$) the input is dominated by the feedback loop
($\mathbf{u} \approx \cont(\mathbf{y})$), and this lack of informativity
makes the direct methods fail in open loop (FIT $-28$ for direct ICI)
despite passing the regulation test. As the excitation grows ($\sigma_r\ge5$), informativity is
restored and the gap between the methods closes.

\emph{(b) Bias of direct methods under colored noise.} Feedback injects
noise into the control input via $\cont(\mathbf v)$, so the regressor loses
orthogonality with the residual and direct identification is
biased (Drawback~\ref{item:D2})~\cite{Forssell99Revisited}. In
Figure~\ref{fig:robot_sweeps} (right, $\sigma_r=10$), as $\sigma_v$ grows
from $0.1$ to $2$ the open-loop FIT of direct ICI collapses
($82.9\!\to\!-7.0$) while indirect ICI degrades gracefully
($86.7\!\to\!60.9$), its regressor $\mathbf r$ being noise-independent by
construction. We examine the Pearson correlation between the model's simulation error and the actual measurement noise. For the $\mathbf u$-driven models, this correlation is significant ($0.28$ for direct ICI, $0.21$ for Direct OE), confirming that they erroneously absorb the noise into the learned dynamics. 

\begin{figure}[t]
\centering
\includegraphics[width=1\linewidth]{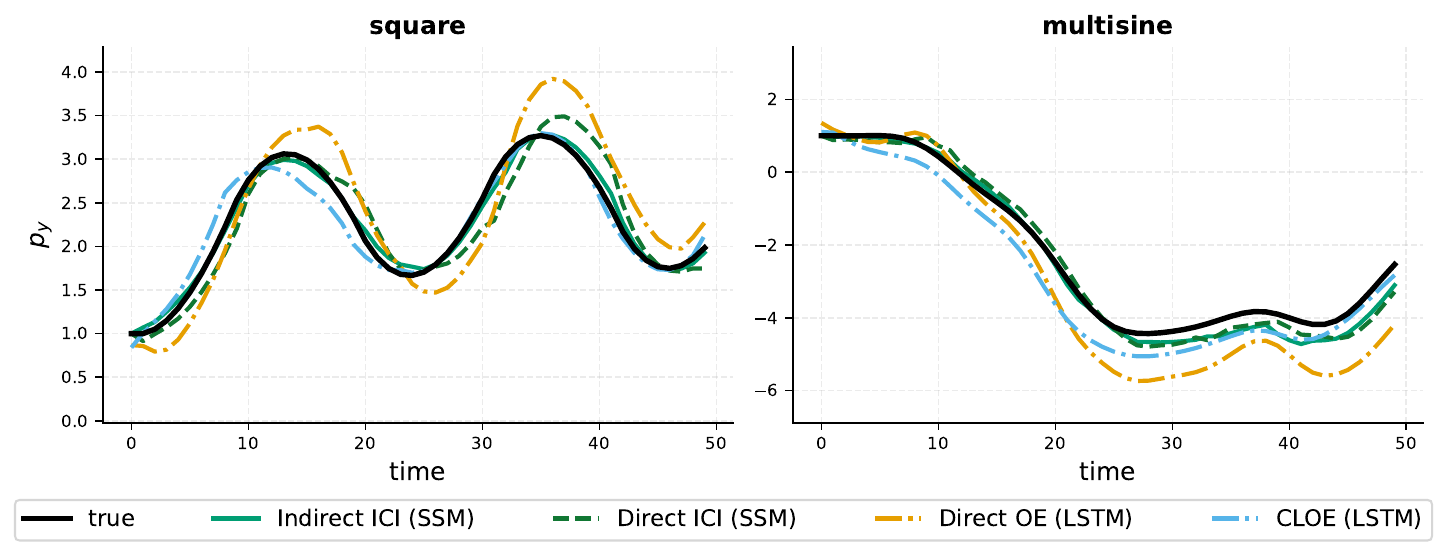}
\caption{Out-of-distribution excitation waveforms (square, multisine) at
the training amplitude; training cell $\sigma_r=20$, $\sigma_v=1$, AR(1);
channel $p_y$ (the analogous $p_x$ figure and all other training cells are
in the repository). Indirect ICI transfers to the unseen shapes more
accurately.}
\label{fig:robot_waveforms}
\end{figure}

\emph{(c) Generalization to out-of-distribution excitation.} To probe
whether the true open-loop dynamics were captured, the identified models are
evaluated on unseen input waveforms (square wave, multisine) at the
training amplitude (Figure~\ref{fig:robot_waveforms}). The indirect ICI model tracks the true response; the
direct methods and the generic baselines exhibit drift and unphysical
oscillations.

\begin{figure}[t]
\centering
\includegraphics[width=1\linewidth]{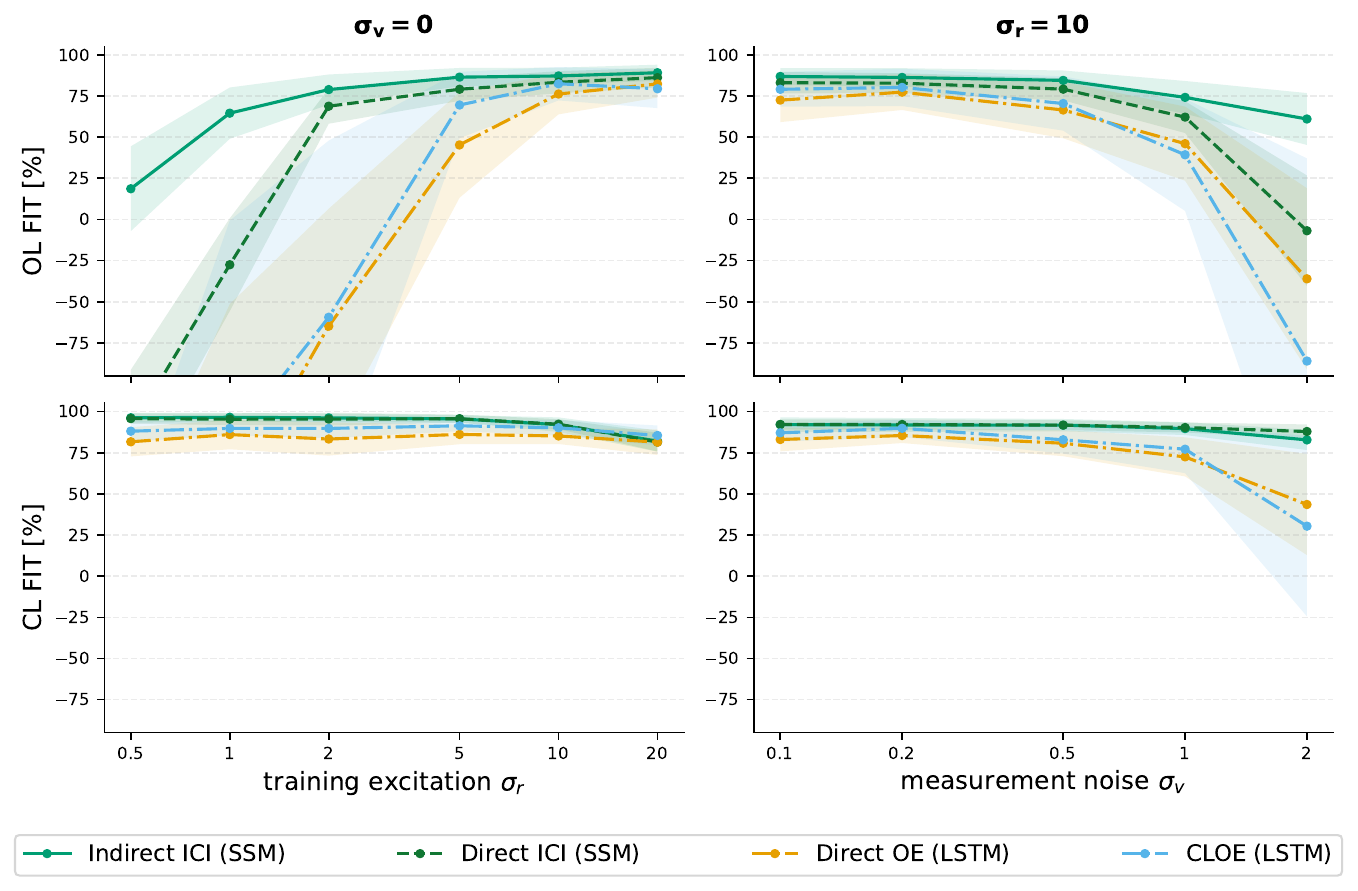}
\caption{Point-mass robot: open-loop (top) / closed-loop (bottom) FIT,
mean$\pm$std over $100$ test trajectories. Left: noise-free training vs.\
excitation $\sigma_r$ (observation~(a)); right: AR(1) noise sweep at
$\sigma_r=10$ (observation~(b)).}
\label{fig:robot_sweeps}
\end{figure}

\subsection{Generalization across environments via context-aware ICI}
\label{subsec:exp_context}

This experiment evaluates the context-enriched parameterization of Section~\ref{sec:context_ici} across a family of environments. We consider the stable point-mass robot of Section~\ref{subsec:setup} with
full-state measurements $y_t = x_t + v_t$ and the proportional controller
on the position components; unlike Section~\ref{subsec:exp_robot}, the
robot now operates in a family of environments with position-dependent
friction: in environment $e$, the nominal friction force is scaled by
$s_e(p) = 1 + \sum_{i=1}^3 w_i \exp\bigl(-\|p - c_i\|_2^2/(2\sigma_i^2)\bigr)$,
a superposition of random ``sticky patches'' with intensities
$w_i \sim \mathcal{U}(2,8)$, radii $\sigma_i \sim \mathcal{U}(2,4)$, and
centers $c_i \sim \mathcal{U}([-9,9]^2)$. Two of these fields are shown in
Figure~\ref{fig:context_fields}: the test trajectories repeatedly traverse the
patches, so the field shapes the dynamics along the whole data record.
 
Two indirect-ICI models share the same certified deep SSM core, with and without access to the local drag scale as context. The contextual model instantiates~\eqref{eq:ctx_split}, with the deep SSM acting as the stable feature generator $\approxOp_p$:
\begin{itemize}
    \item \emph{SSM-no-context} receives $\excMat$ without terrain information;
    \item \emph{SSM-context} uses a context-only mixer producing $A_t\in\R^{\dimOut\times d_f}$, applied to the core features and rescaled so that $\Vert A_t\Vert_F\leq1$. This also bounds its induced Euclidean norm, as required by Section~\ref{sec:context_ici}.
\end{itemize}
During training, the context sequence is refreshed once per epoch by evaluating the terrain map at the model's predicted positions. During deployment, it is recomputed sequentially as $\hat\zeta_t=s_e(\hat p_t)$, where $\hat p_t$ is the position component of the predicted output. The mixer uses context with a one-step delay, so $\hat\zeta_t$ affects predictions only from step $t+1$, preserving strict causality. The terrain map is assumed known, but no measured or recorded test trajectory is used to construct context.

Two training regimes share a budget of $60$ trajectories: for Case~A the training dataset contains trajectories collected on a single field, while in Case~B the training trajectories are collected from $12$ fields; both are tested on the training fields and on $6$ held-out unseen fields.

\textbf{Results.} 
The open-loop identification accuracy is illustrated in Figures~\ref{fig:context_candles}--\ref{fig:context_channels}, with comprehensive per-environment tables available in the public repository. Bound~\eqref{eq:ctx_bound}, together with~\ref{item:T1}, certifies stability of the identified closed-loop models. All evaluated closed-loop rollouts remained bounded across models, seeds, and fields. From these experiments, two primary conclusions emerge:
 
\emph{(a) Context distinguishes environment-specific responses.} The initial state and excitation $(x_0,\excMat)$ do not specify the environment, so a predictor receiving only these inputs must fit a common response across fields. Under family training, the no-context model attains similar open-loop FIT on seen and unseen fields ($42.4\%$ and $42.2\%$). Increasing its parameter count from $2\,793$ to $15\times$ ($42\,249$) does not improve unseen-field FIT ($36.9$--$40.7\%$), consistent with a limitation due to missing environment information. Figure~\ref{fig:context_channels} shows the resulting prediction errors inside the strong-friction patches.

\emph{(b) Context enables generalization to unseen fields without retraining.} Under family training, the contextual model reaches $86.6\pm6.7\%$ unseen-field FIT against $42.2\pm24.5\%$ without context, an improvement of $44.4$ percentage points with a smaller reported standard deviation. Its unseen-field score is close to the seen-field score ($86.6\%$ versus $87.0\%$), showing generalization to the tested unseen environments. This transfer uses the same learned parameters, with each new terrain map available for context evaluation.

\begin{figure}[t]
\centering
\includegraphics[width=.82\linewidth]{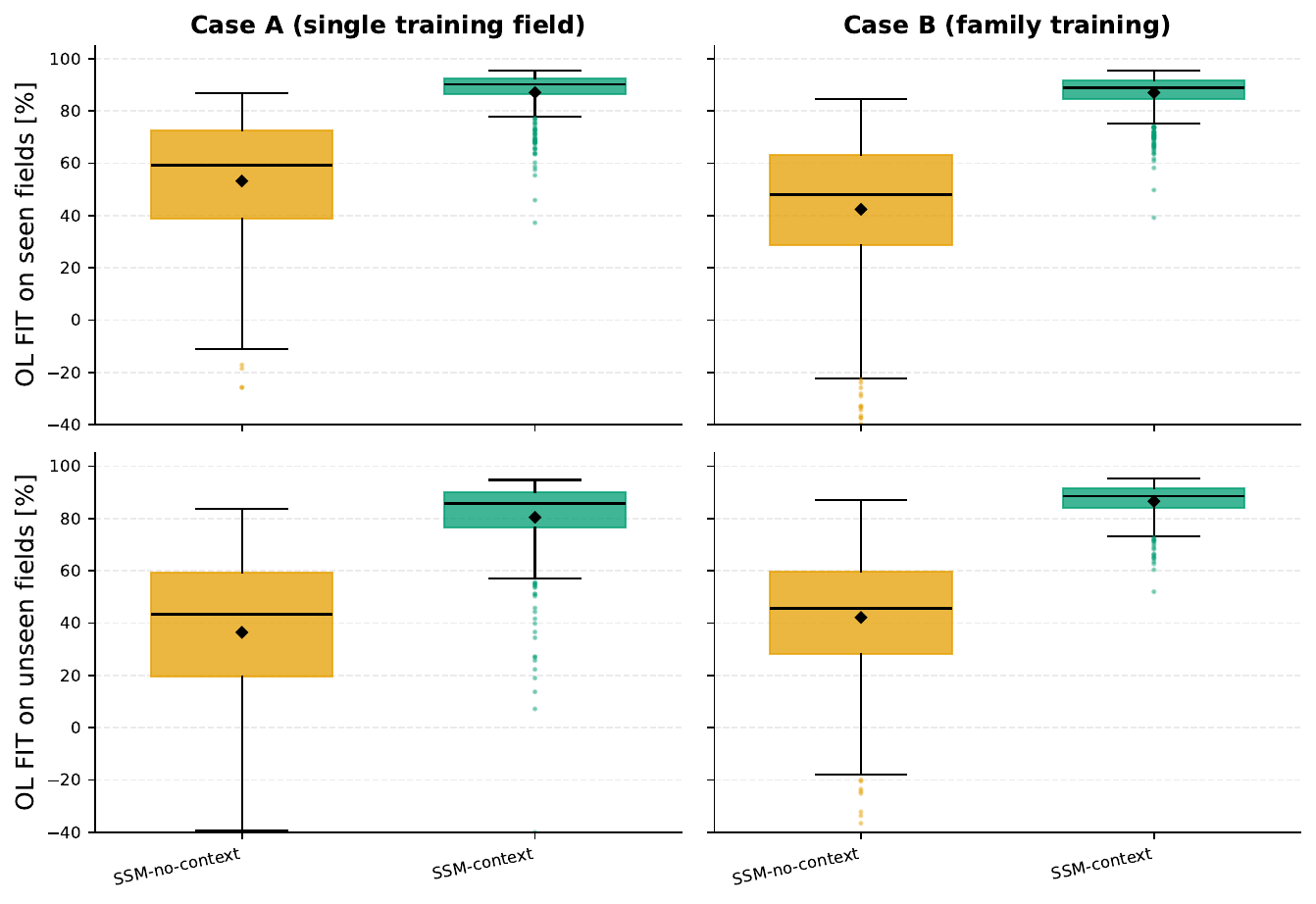}
\caption{Distribution of the open-loop per-trajectory FIT for the two models (box $=$ median and interquartile range, whiskers $=$ Tukey, $\diamond=$ mean; all seeds and test trajectories pooled). Columns vary the training regime (Case A, a single field~/~Case B, a family of fields) and rows the test set (seen
fields, top~/~unseen fields, bottom); the $y$-axis is FIT in percent. The contextual model dominates the uncontextualized one in every panel, and under family training its seen- and unseen-field distributions nearly coincide.}
\label{fig:context_candles}
\end{figure}

\begin{figure}[t]
\centering
\includegraphics[width=.94\linewidth]{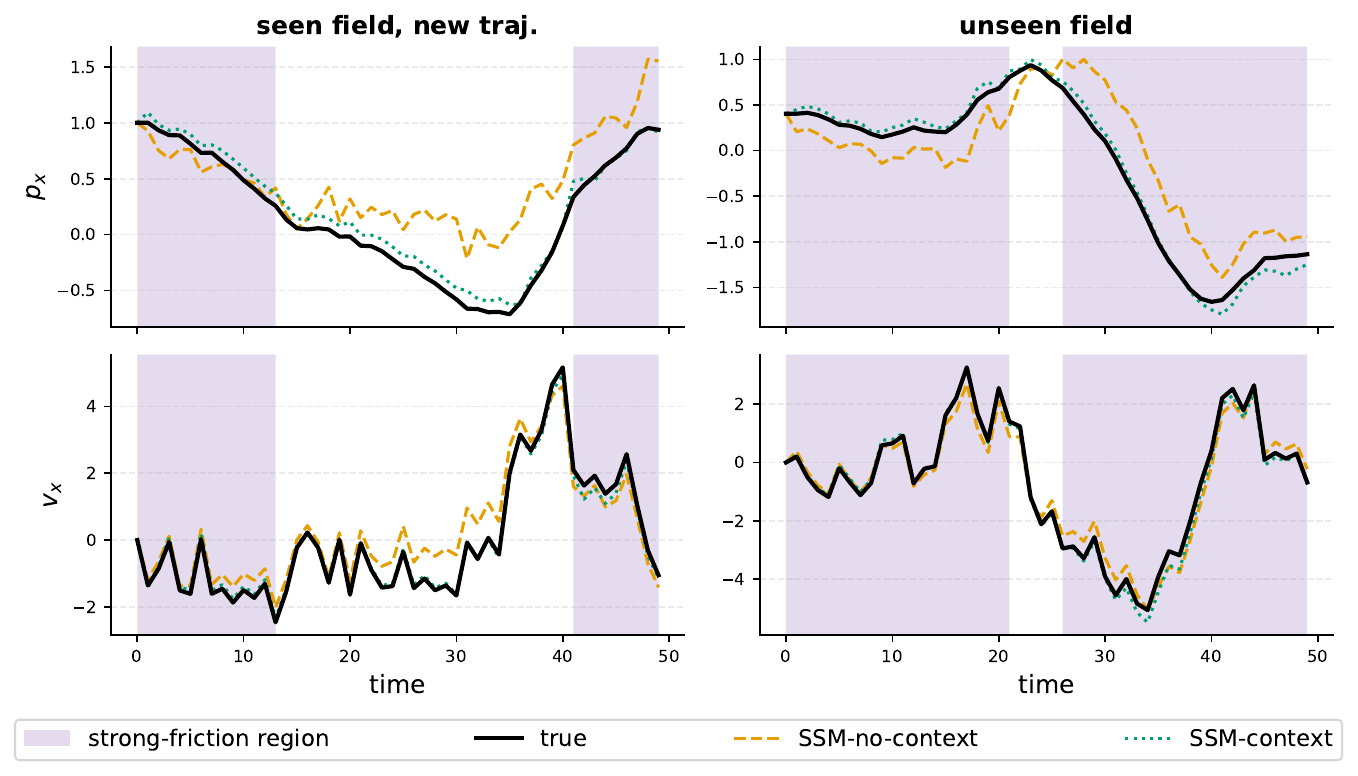}
\caption{Case~B, open-loop response on a seen (left) and unseen (right) field: $p_x$ (top), $v_x$ (bottom); shaded $=$ strong friction ($s_e>2.5$). The no-context prediction departs from the truth inside the patches, while the contextual model tracks it; the analogous $p_y$/$v_y$ figures are in the public repository.}
\label{fig:context_channels}
\end{figure}

\begin{figure}[t]
\centering
\includegraphics[width=1\linewidth]{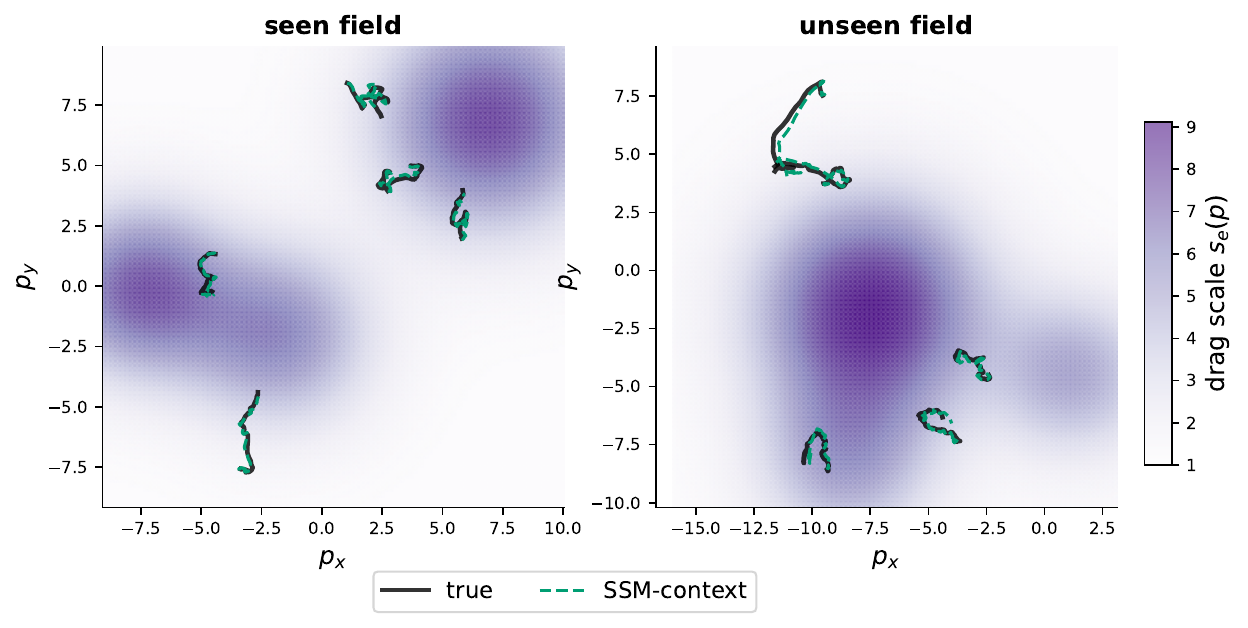}
\caption{Case B: true open-loop test trajectories (black) and SSM-context rollouts (dashed) over the friction field $s_e(p)$ (darker $=$ stickier), on
a seen (left) and unseen (right) field, under the causal predicted-position
context. Dots mark the initial positions.}
\label{fig:context_fields}
\end{figure}

\subsection{PAC-Bayesian Identification under Data Scarcity}
\label{subsec:exp_pacbayes}

This experiment considers the regime in which the PAC-Bayesian framework of Section~\ref{sec:PAC} is most useful, namely scarce and noisy data with no validation set, where ERM does not merely add variance but memorizes the measurement noise. 

On the point-mass robot example, we demonstrate three key properties: (a)~under data scarcity ERM overfits and destroys a good prior, whereas the SVGD approximation of the Gibbs posterior over the same ICI parameterization generalizes and preserves it; (b)~Theorem~\ref{theo:pac-id} provides a non-vacuous generalization certificate, spending no held-out data, that tightens with the sample size $N$; and (c)~this certificate requires only an informative prior, which may be supplied by cross-domain transfer or, in the absence of a source system, obtained from data split.

In this experiment, we apply the transformation in~\eqref{eq:loss_trans} to the identification loss, setting $C=1$ and $\gamma=10$: every risk bound then takes values in $[0,1)$ and is vacuous as soon as it reaches $1$. 
The identified operator and every posterior particle are contractive RENs, certified closed-loop stable by Theorem~\ref{theo:ici}. Four estimators share the same parameterization: the un-adapted prior mean, ERM from scratch, ERM fine-tuned from the prior and one particle drawn uniformly from the SVGD approximation of the Gibbs posterior~\eqref{eq:qstar-id} (Section~\ref{subsec:numerical_SVGD}). The prior parameters $\theta_{\mathrm{pre}}$ are obtained by ERM pretraining on~\eqref{eq:indirect} using trajectories from a nominal plant with uniform linear drag; at deployment the plant has drifted to a heavier drag, modeling wear, and must be re-identified from only $N \in \{1, \dots, 64\}$ noisy trajectories, with no validation set. Being trained on a different plant, $\theta_{\mathrm{pre}}$ is independent of the target sample, so Theorem~\ref{theo:pac-id} applies with no data splitting.
The prior distribution is the per-parameter Gaussian
$\mathcal N\big(\theta_{\mathrm{pre},i},(\alpha|\theta_{\mathrm{pre},i}|+10^{-3})^2\big)$.

\textbf{Results.} Figures~\ref{fig:pac_bound}--\ref{fig:pac_fit} report three complementary quantities. Figure~\ref{fig:pac_bound} (left) shows the PAC-Bayes certificate~\eqref{eq:bound_qstar-id} of a single draw from the exact Gibbs posterior $\Q^*$; Figure~\ref{fig:pac_bound} (right) shows the empirical risk~\eqref{eq:empirical_loss} computed on $c$ rather than on the transformed loss $\Tilde{c}$~\eqref{eq:loss_trans}, against noisy outputs; and Figure~\ref{fig:pac_fit} shows deployment FIT against the noise-free plant. Only the first quantity is certified by Theorem~\ref{theo:pac-id}; the latter two are complementary empirical diagnostics. The complete numerical table is provided in the repository. The analysis yields three main observations.

\emph{(a) At scarce data ERM memorizes the noise; the SVGD estimator does not.} For $N \le 4$ both ERM models drive their training MSE below the
noise floor $\sigma_v^2$, fitting the specific realization, while their
test error explodes and their deployment FIT drops \emph{below} the initial pre-trained prior ($67.7\%$ vs.\ $86.6\%$ at $N=4$;
Figure~\ref{fig:pac_fit}). Fine-tuning from the good prior overfits
just the same, isolating the benefit as the KL anchoring
of~\eqref{eq:def_qstar} rather than a good initialization: the SVGD estimator, optimizing the same trajectories, keeps its training error at the
floor and its FIT at the prior's level, and is unaffected even at $N=1$.

\emph{(b) The SVGD estimator performs well at scarce data; the Gibbs bound becomes non-vacuous.} Across $N$ (Figure~\ref{fig:pac_fit}) the posterior
holds an almost constant ${\approx}85\%$ FIT, whereas the ERM starts far below
the prior, catches up around $N=8$, and overtakes only as the sample size grows. This plateau is the price of anchoring to the prior, and it is
exactly what protects the posterior where unconstrained fitting ruins ERM.
The certificate~\eqref{eq:bound_qstar-id} behaves consistently
(Figure~\ref{fig:pac_bound}): vacuous at the smallest $N$, it becomes
informative from $N=4$ and tightens with more data, at the $\sqrt N$ rate of
$\lambda^\ast$; the point where it stops being vacuous thus marks the onset
of the useful regime. The bound is fully computed from data (no oracle
term) and it holds with probability $\ge 1-\delta$ ($\delta=0.1$). Every SVGD particle is certified closed-loop stable by Theorem~\ref{theo:ici}, whereas the loss certificate applies to an exact draw from $\Q^*$. The approximation error for the deployed particle is not covered (Section~\ref{subsec:numerical_SVGD}). Neither guarantee uses held-out data.

The prior width $\alpha$ trades deployment fit against certification: with a
wide prior (Figure~\ref{fig:pac_fit}, pink) the posterior detaches from
$\theta_{\mathrm{pre}}$ and reaches a higher FIT, but a wide prior rarely
samples low-loss models, so the bound turns vacuous, which is precisely why the selected configuration keeps $\alpha$ small.

\emph{(c) The informative prior can be built without a source system.} When
no related plant is available for transfer, an equally good prior is
obtained from the target data alone, by splitting it in two
(Prop.~2 of~\cite{PACSNOC}): a short ERM run on one half provides the prior
center, and the posterior and its certificate are computed on the other
half. This route reaches a comparable FIT (${\approx}93\%$) with a valid
certificate (Figure~\ref{fig:pac_bound}), at the cost of needing more data
than transfer, since half the sample is spent on the prior. Either way, an
informative prior is essential: a non-informative one forfeits the
certificate.

\begin{figure}[t]
\centering
\begin{minipage}{0.5\linewidth}
\centering
\includegraphics[width=\linewidth]{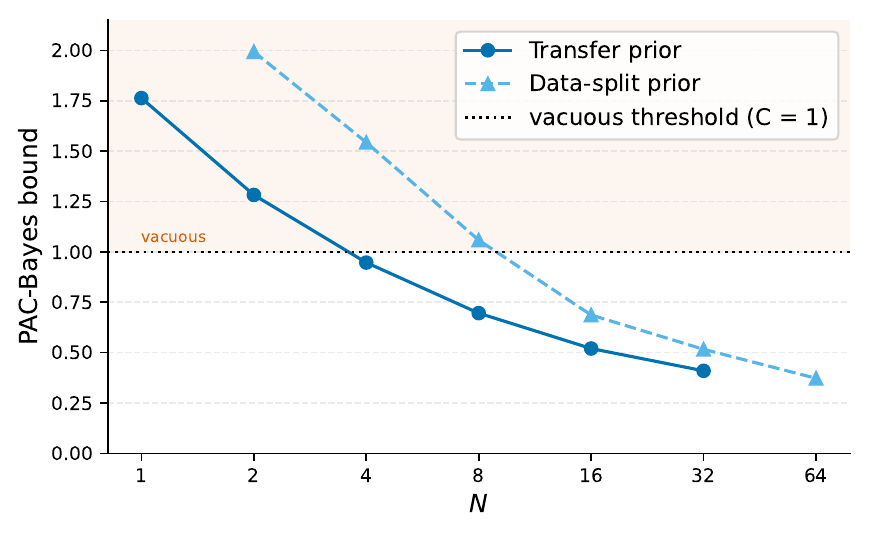}
\end{minipage}\hfill
\begin{minipage}{0.5\linewidth}
\centering
\includegraphics[width=\linewidth]{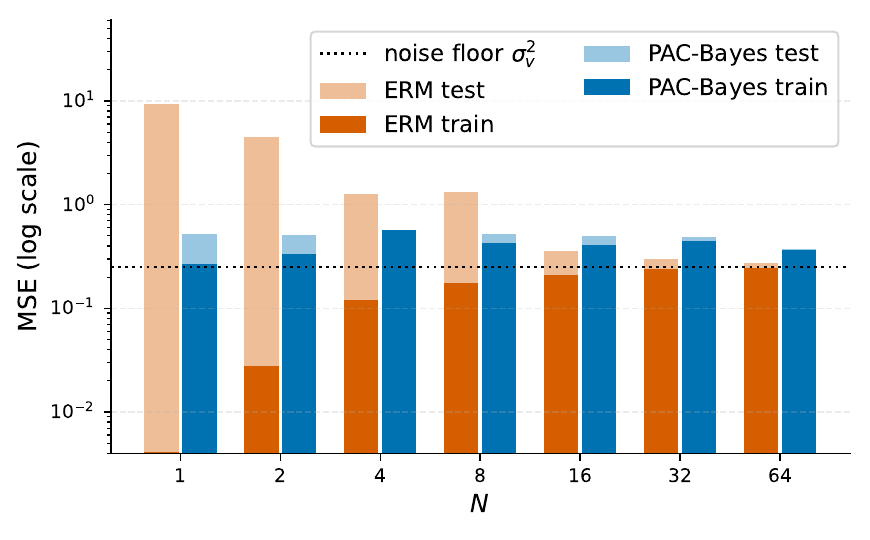}
\end{minipage}

\caption{Certificate and overfitting diagnostics vs.\ sample size $N$.
\emph{Left:} the~\eqref{eq:bound_qstar-id} bound for the transfer prior
(blue) and the transfer-free data-split prior (light blue); the shaded band
marks the vacuous region ($C=1$). \emph{Right:} train (dark) vs.\ test
(light) empirical risk~\eqref{eq:empirical_loss} built on
the untransformed loss~\eqref{eq:mse_loss} for the ERM
(orange) and PAC-Bayes (blue) models, against noisy outputs (log scale).
The left panel bounds the true risk in~\eqref{eq:true_loss} for a draw from $\Q^*$; the right is an empirical diagnostic.
The MSE noise floor is $\sigma_v^2=0.25$ (dotted), and training below it indicates noise memorization.}

\label{fig:pac_bound}
\end{figure}

\begin{remark}[Estimating the log-partition]\label{rem:Z}
The certificate~\eqref{eq:bound_qstar-id} requires
$Z_\lambda = \mathbb{E}_{\theta\sim\P}\,e^{-\lambda \hat J(\theta,\dataset_T^N)}$,
an expectation under the prior with no closed form. We estimate it by Monte
Carlo, drawing $M = 1024$ i.i.d.\ prior samples $\theta_j\sim\P$ and averaging
$e^{-\lambda \hat J(\theta_j,\dataset_T^N)}$. Since the integrand lies in
$(0,1]$, a one-sided Hoeffding bound on this average controls the estimation
error, so that $-\tfrac1\lambda\ln Z_\lambda$ is replaced by a
\emph{high-probability upper bound}; the resulting certificate remains valid,
now jointly over the sample, the posterior draw, and the Monte-Carlo draws.

\end{remark}

\begin{figure}[t]
\centering
\includegraphics[width=1\linewidth]{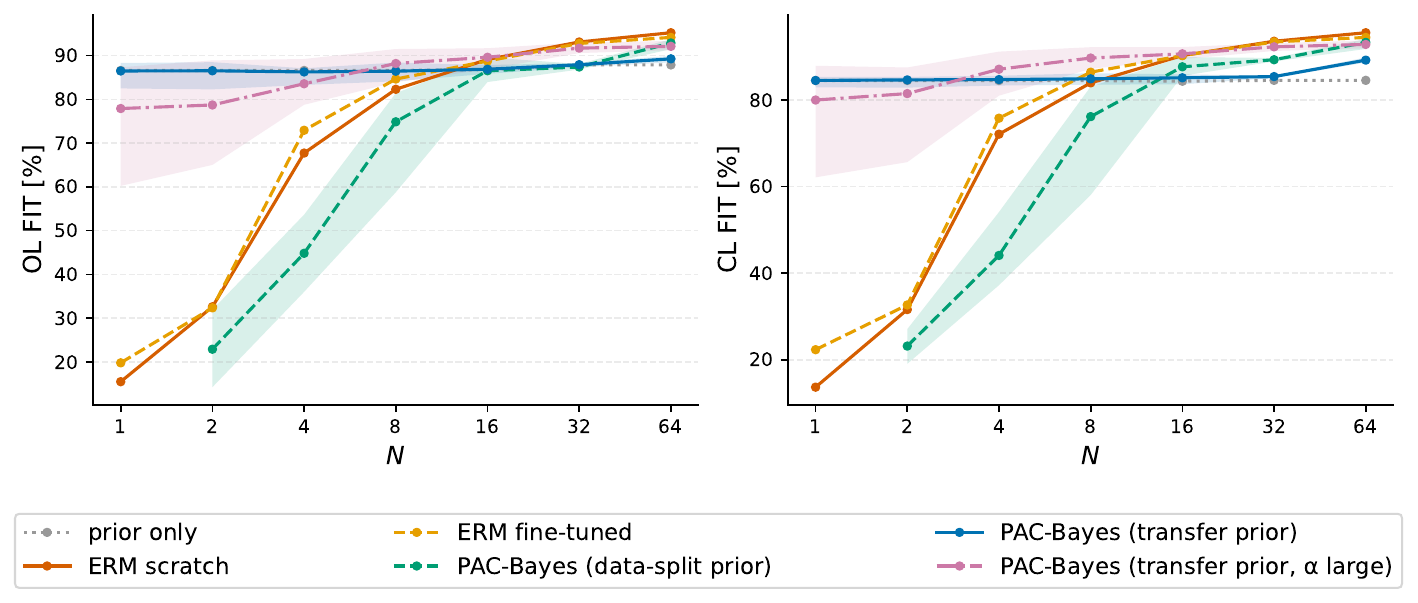}
\caption{Deployment FIT vs.\ $N$ (mean over seeds): left, open loop; right, closed loop;
against the noise-free plant. Key comparisons: 1) ERM methods (orange/gold) overfit to noise under data scarcity, degrading performance below the baseline prior. 2) The \emph{SVGD estimator} (blue), anchored to the \emph{transfer prior}, resists overfitting and preserves the prior's quality at scarce $N$. ERM overtakes it only at larger sample sizes. 3) Relying solely on target data, the \emph{data-split prior} approach (green) adapts well and maintains a valid certificate. 4) Using a \emph{wide prior} (pink) yields high FIT but results in a vacuous certificate.}
\label{fig:pac_fit}
\end{figure}

\section{Conclusions and future work}
This paper presented a computationally tractable framework for closed-loop
identification of nonlinear systems.
By leveraging a complete parameterization of all systems stabilized by a given controller, the proposed Internal Controller Identification (ICI) framework guarantees closed-loop stability of the identified model by construction. Moreover, this approach ensures that training remains well posed even for open-loop unstable plants, and reduces to a consistent open-loop problem in
the linear case.

Our numerical experiments validate these theoretical claims across several distinct operating regimes.
For open-loop unstable plants, the certified models successfully captured the unstable dynamics without ever destabilizing the loop, avoiding the divergence seen in uncertified baselines. On a partially observed robotic system, the indirect criterion proved to be the most accurate under conditions of weak excitation and feedback-correlated noise.
We also showed that bounded context wrappers allow a single certified model to identify an entire family of environments, generalizing even to members unseen during training. Finally, in data-scarce regimes, the PAC-Bayesian extension prevented the noise memorization that typically degrades ERM, while providing non-vacuous certificates without requiring held-out data.

Future work will focus on applying these identified models to adaptive control schemes and designing an iterative procedure that seamlessly integrates system identification with controller synthesis.

\bibliographystyle{abbrv}        %
\bibliography{root}           %

\end{document}